\documentclass[11pt]{article}
\usepackage{bm,color,url}
\usepackage{amssymb,float}
\usepackage{enumerate}
\usepackage{amsmath}
\usepackage{natbib}
\usepackage{graphicx}
\usepackage{subfigure}
\usepackage{bbm}
\usepackage[margin=1.1in]{geometry}
\usepackage{adjustbox}
\usepackage{multirow}

\newtheorem{theorem}{Theorem}

\newtheorem{remark}[theorem]{Remark}
\newtheorem{example}[theorem]{Example}
\newtheorem{proposition}[theorem]{Proposition}
\newenvironment{proof}[1][Proof]{\textbf{#1.} }{\ \rule{0.5em}{0.5em}}

\newenvironment{mat}{\left[\begin{array}{ccccccccccccccc}}{\end{array}\right]}
\newcommand\bcm{\begin{mat}}
\newcommand\ecm{\end{mat}}

\newenvironment{rmat}{\left[\begin{array}{rrrrrrrrrrrrr}}{\end{array}\right]}
\newcommand\brm{\begin{rmat}}
\newcommand\erm{\end{rmat}}

\newcommand{\EE}{{\mathord{I\kern -.33em E}}}
\def\E{{\EE}} 
\def\Q{{\mathbb Q}} 
\def\R{{\mathbb R}} 
\def\1{{\mathbf 1}} 
\def\F{{\mathcal F}} 
\def\P{{\mathbb P}}

\title{\LARGE \bf  Optimal Dynamic Futures Portfolio in a Regime-Switching Market Framework}

\author{Tim Leung 
\thanks{Department of Applied Mathematics, University of Washington, Seattle WA 98195. E-mail:
\mbox{timleung@uw.edu}. Corresponding author. } \and Yang Zhou\thanks{Department of Applied Mathematics, University of Washington, Seattle WA 98195. E-mail:
\mbox{yzhou7@uw.edu}.} }
\begin{document}

\maketitle

\begin{abstract} 
We study the problem of dynamically trading futures in a regime-switching market. Modeling the underlying asset price as a Markov-modulated diffusion process, we present a utility maximization approach to determine the optimal futures trading strategy. This leads to the analysis of the associated system of Hamilton-Jacobi-Bellman (HJB) equations, which are reduced to a system of linear ODEs.  We apply our stochastic framework to two models, namely, the Regime-Switching Geometric Brownian Motion (RS-GBM) model and Regime-Switching Exponential Ornstein-Uhlenbeck (RS-XOU) model. Numerical examples are provided to illustrate the investor's optimal futures positions and portfolio value across market regimes. 
\end{abstract}

\newpage
\section{Introduction}
Futures are bilateral contracts of agreement to buy or sell an asset at a pre-determined price at a pre-specified time in the future. The underlying assets can be physical commodities, market indexes, or financial instruments. Futures are standardized and exchange-traded. The Chicago Mercantile Exchange (CME), which is the world's largest futures exchange, averages over 15 million futures contracts traded per day.\footnote{Source: CME Group daily exchange volume and open interest report, available at \url{https://www.cmegroup.com/market-data/volume-open-interest/exchange-volume.html}} Within the universe of hedge funds and alternative investments, futures funds play an integral role with hundreds of billions under management. This motivates us to investigate the problem of trading  futures portfolio dynamically over time.
 
In this paper, we present a stochastic control approach to generate dynamic futures trading strategies. We consider a general regime-switching framework in which  the stochastic market regime is represented by a continuous-time finite-state Markov chain. The underlying asset's spot  price is modeled by a Markov-modulated diffusion process.  Then, we derive the no-arbitrage price dynamics for the futures contracts and determine the optimal futures trading strategy by solving a utility maximization problem. By analyzing and solving the associated    Hamilton-Jacobi-Bellman (HJB) equations, we derive the investor's value function and optimal trading strategies. 

Among our findings, we show that the investor's value function admits a   separable form under a general regime-switching market framework, and the original HJB equations are reduced to a system of linear ODEs.  In addition, we also define the investor's certainty equivalent to  quantify the value of the futures trading opportunity to the investor.  Surprisingly both the value function and certainty equivalent do not depend on the current spot and futures prices, and admits a universal form across different model specifications. Nevertheless, the risk premia associated with the  regime-switching model play a crucial role, not only in the  value function and certainty equivalent, but also in  the optimal strategy. In addition, we show two applications of our model, Regime-Switching Geometric Brownian Motion (RS-GBM) and Regime-Switching Exponential Ornstein-Uhlenbeck (RS-XOU) model. Exponential Ornstein-Uhlenbeck model is one of the mean reversion models used to analyze futures market.


In the literature, the notion of regime switching has been widely applied to  derivatives pricing, such as \cite{Buffington2002} and \cite{Elliott2005}, among many.  The stochastic control approach has been applied to stock portfolio optimization dating  back to \cite{Merton1971}, but much less has been done for dynamic futures portfolio in continuous time. In a  number of   companion papers  \citep{LeungYan2018,LeungYan2019,BahmanLeung,BahmanLeung2}, the  utility maximization approach is used to derive  dynamic   futures  trading strategies under various stochastic models without regime switching.  Also, we refer to  \cite{ZhouYin2003,Leung_regime2010,Chen19} for continuous-time portfolio optimization with regime switching. In comparison to these studies,  the current paper investigates the trading of derivatives (futures), instead of stocks, in a general regime-switching market framework that can be applied to an array of regime-switching models. This different setup leads to the interaction between the historical measure and risk-neutral pricing measure, and how they affect the evolution of futures prices and portfolio wealth.  In particular, the regime-switching feature also leads to jumps in the futures prices and the investor's wealth process, even though the underlying asset price has continuous paths. Moreover, under our regime-switching framework, the certainty equivalent is the same across different underlying models as long as the market price of risk stays the same.  



The rest of this paper is structured as follows. We describe the general market framework in Section \ref{sect-fut-port}. The dynamic futures portfolio optimization is discussed in Section \ref{sect-utility}. Then, we apply our framework to the RS-GBM model and RS-XOU model in Sections \ref{sect-GBM} and \ref{sect-EXP-OU}, respectively. Concluding remarks are provided in Section \ref{sect-conclude}.




\section{Futures Price  Dynamics}\label{sect-fut-port}

We fix a probability space $(\Omega,\mathcal{G},\mathbb{Q})$, where $\mathbb{Q}$ is the risk-neutral pricing measure $\Q$.  Let $\xi$ be a continuous-time irreducible finite-state Markov chain with state space $E=\{1,2,\ldots,M\}$.
 The generator matrix of $\xi$, denoted by $\bm{\tilde Q}$,  has entries $\bm{\tilde Q}(i,j)=\tilde q(i,j)$  such that $\tilde{q}(i,j)\ge 0$ for $i\neq j$ and $\sum_{j\in E}\tilde{q}(i,j)=0$ for $i\in E$. This Markov chain represents the changing market regime   and  influences the underlying asset's price dynamics.
 
We can use a stochastic integral with respect to a Poisson random measure
to represent Markov chain $\xi$. For $i,j\in E$ with $i\neq j$, let $\Delta(i,j)$ be the consecutive left-closed,
right-open intervals of the real line, each having length $\tilde{q}(i,j)$. Define a function $h:E\times\R\rightarrow\R$ by 
\begin{equation}
h(i,z)=\sum_{j\in E\setminus\{i\}}(j-i)I_{\{z\in\Delta(i,j)\}}.
\end{equation}
Then, under measure $\Q$, the Markov chain $\xi_t$ evolves according to 
\begin{equation}
d\xi_t =\int_\R h(\xi_t,z)N(dt,dz),\label{Xi_SDE_N}
\end{equation}
where  $N(dt,dz)$ is the Poisson random measure with  intensity $dt\times\tilde{\mu}(dz)$ and $\tilde{\mu}$ is the Lebesgue measure satisfying
\begin{equation}
\int_\R I_{\{z\in\Delta(i,j)\}}\tilde{\mu}(dz) = |\Delta(i,j)|=\tilde{q}(i,j),\label{LebesgueMeasure}
\end{equation}
with $|\Delta(i,j)|$ being the length of $\Delta(i,j)$.  We can also express  SDE \eqref{Xi_SDE_N}  as
\begin{align}
d\xi_t = \sum_{j\in E\backslash\{\xi_t\}}\tilde{q}(\xi_t,j)(j-\xi_t)dt+\int_\R h(\xi_t,z)M^\Q(dt,dz),
\end{align}
using the compensated Poisson process under measure $\Q$ defined by
\begin{equation}
M^\Q(dt,dz)=N(dt,dz)-dt\times\tilde{\mu}(dz).
\end{equation}

The underlying asset's spot price is denoted by $S_t$. Its log-price,   $X_t=\log(S_t)$, evolves according to 
\begin{equation}
dX_t = \tilde a(t,X_t,\xi_t)dt+b(t,X_t,\xi_t)dZ_t^\Q,\label{log-price}
\end{equation}
where $Z^\Q$ is the standard Brownian motion under the measure $\Q$ and independent of $\xi$. The drift and volatility functions  $\tilde{a}(\cdot, \cdot,\cdot)$ and $b(\cdot, \cdot,\cdot)$ are assumed to satisfy the conditions such that  SDE\eqref{log-price} has a strong solution.  

Consider a futures contract on  the underlying asset $S$ with maturity $T$. Its no-arbitrage price at time $t\le T$ is given by  the conditional expectation  under the risk-neutral pricing measure $\mathbb{Q}$:
\begin{equation}
F_i(t,x) = \E^\Q[\exp(X_{T})|X_t=x,\xi_t=i].
\end{equation}
The futures price function $F_i(t,x)$  is determined from the  following  system of PDEs 
\begin{equation}
\begin{aligned}
\partial_t F_i+\mathcal{L}^\Q_i F_i
+\sum_{j\in E\setminus\{ i\}} \tilde q(i,j)( F_j- F_i)=0,\label{FuturePDE}
\end{aligned}
\end{equation}
for $(t,x)\in[0,T)\times\R$ and $i=1,\ldots,M$, where 
\begin{equation}
\mathcal{L}^\Q_i \cdot := \tilde a(t,x,i)\partial_x\cdot+\frac{b^2(t,x,i)}{2}\partial_{xx}\cdot\,.
\end{equation}
 To facilitate presentation, we have dropped the variables from different functions in   (\ref{FuturePDE}) and will do the same in PDEs that follow when no ambiguity arises. 
 

For the futures trading problem,  asset and futures prices are observed under the physical measure $\P$. Under measure $\mathbb{P}$, the Markov chain $\xi$ has generator matrix $\bm{Q}$ with  entries $\bm{Q}(i,j)=q(i,j)$, where $i,j\in E$. Since $\mathbb{P}$ and $\Q$ are equivalent measures, we have $q(i,j)=0$ iff $\tilde q(i,j)=0$. To relate the Poisson random measures under measures $\P$ and $\Q$, we denote by $\mu(dz)$ the intensity measure of $N(dt,dz)$ under measure $\P$ such that
\begin{equation}
 \mu(dz) = \left\{
 \begin{aligned}
 &\frac{q(i,j)}{\tilde{q}(i,j)}\tilde{\mu}(dz),&\mbox{for }z\in\Delta(i,j),\\
 & \tilde{\mu}(dz), &\mbox{others},
 \end{aligned}
 \right.
 \end{equation}
 under the convention that $0/0=1$.   Then, the compensated Poisson process under measure $\P$ is
\begin{equation}
\begin{aligned}
&M^\P(dt,dz) = M^\Q(dt,dz)-\sum_{i,j\in E,i\neq j}\frac{q(i,j)-\tilde{q}(i,j)}{\tilde{q}(i,j)}I_{\{z\in \Delta(i,j)\}}dt\times\tilde{\mu}(dz).
\end{aligned}
\end{equation}
Accordingly, the Markov chain $\xi_t$ satisfies
\begin{align}
d\xi_t =  \sum_{j\in E\backslash\{\xi_t\}}q(\xi_t,j)(j-\xi_t)dt+\int_\R h(\xi_t,z)M^\P(dt,dz).
\end{align}

 To relate the Brownian motions under measures $\P$ and $\Q$, we denote by $\zeta(\xi_t)$ the  risk premium associated with the Brownian motion such that 
\begin{equation} \label{ZPZQ}
dZ_t^\Q=dZ_t^\P+\zeta(\xi_t) dt,
\end{equation}
In turn, the log spot price   satisfies
\begin{equation}
dX_t=a(t,X_t,\xi_t)dt+b(t,X_t,\xi_t)dZ_t^{\P},
\end{equation}
whose drift is given by 
\begin{align*}
a(t,x,i) := \tilde{a}(t,x,i)+ \zeta(i) b(t,x,i).
\end{align*}  
Here, $Z^\P$ is the standard Brownian motion under $\P$ and is  independent of $\xi$.

Applying It\^{o}'s lemma, the futures price $F_t$ satisfies the stochastic differential equation (SDE),
\begin{equation}
\begin{aligned}
dF_t &= \eta(t,X_t,\xi_t)dZ_t^\Q+\int_\R \sum_{j\in E\setminus\{\xi_t\}} \Delta_F(t,X_t,\xi_t,j)I_{\{z\in\Delta(\xi_t,j)\}}M^\Q(dt,dz)\label{FutureSDE},
\end{aligned}
\end{equation}
where we have defined  
\begin{align}
\eta(t,x,i)&=b(t,x,i)\partial_x F_i (t,x),\\
\Delta_F(t,x,i,j)&=F_j(t,x)-F_i(t,x),\label{Difference_Function_N_Regimes}
\end{align} 
for $i,j\in E$. In particular, $\Delta_F(t,x,i,i)=0$ by definition. We note that $F_t$ is a $\Q$-martingale. Applying \eqref{ZPZQ}, the  futures price   $F_t$ admits the   $\P$-dynamics:
\begin{equation}\label{futuresPSDE}
\begin{aligned}
dF_t&=\bigg(\eta(t,X_t,\xi_t)\zeta(\xi_t)+\sum_{j\in E\setminus\{\xi_t\}}(q(\xi_t,j)-\tilde{q}(\xi_t,j))\Delta_F(t,X_t,\xi_t,j) \bigg)dt\\
&+\eta(t,X_t,\xi_t)dZ_t^{\P}+\int_\R \sum_{j\in E\setminus\{\xi_t\}} \Delta_F(t,X_t,\xi_t,j)I_{\{z\in\Delta(\xi_t,j)\}}M^\P(dt,dz).
\end{aligned}
\end{equation}
A key feature of this regime-switching framework is that the futures price process is a jump-diffusion even though the  spot price process has continuous paths. The jumps in the futures prices will have direct impact on the   strategies in dynamic futures portfolio.


\section{Futures Portfolio Optimization}\label{sect-utility}
We now consider the problem of dynamic trading futures contracts. Consider $M$  futures $\bm{F}=(F^{(1)},\ldots,F^{(M)})$ written on the same   asset $S$ with different maturities, denoted by $T_1<T_2<\ldots<T_M$ without loss of generality. We define $M\times M$  \textit{coefficient matrix} by
 \begin{equation}\label{coefficient_Definition}
 \bm{\Gamma}(t,x,i)= \bcm\eta^{(1)}(t,x,i)&\eta^{(2)}(t,x,i)&\cdots& \eta^{(M)}(t,x,i)\\
 \Delta_F^{(1)}(t,x,i,1)&\Delta_F^{(2)}(t,x,i,1)&\cdots&\Delta_F^{(M)}(t,x,i,1)\\
 \vdots& \vdots&\vdots &\vdots\\
 \Delta_F^{(1)}(t,x,i,i-1)&\Delta_F^{(2)}(t,x,i,i-1)&\cdots&\Delta_F^{(M)}(t,x,i,i-1)\\
 \Delta_F^{(1)}(t,x,i,i+1)&\Delta_F^{(2)}(t,x,i,i+1)&\cdots&\Delta_F^{(M)}(t,x,i,i+1)\\
  \vdots& \vdots&\ddots &\vdots\\
 \Delta_F^{(1)}(t,x,i,M)&\Delta_F^{(2)}(t,x,i,M)&\cdots&\Delta_F^{(M)}(t,x,i,M) \ecm.
\end{equation}
 It follows from \eqref{FutureSDE} that if this $M\times M$ matrix $\bm{\Gamma}$ is invertible, then it is sufficient to these $M$ futures     to  fully replicate any other futures on $S$ with a different maturity, up to the shortest maturity. To see this, for a futures with an arbitrary maturity $T$, its price is connected with the prices of the $M$ futures as follows:
\[
dF_t\! =\bcm dF_t^{(1)}&dF_t^{(2)}&\cdots&dF_t^{(M)} \ecm \bm{\Gamma}(t,X_t,\xi_t)^{-1}\bcm\eta(t,X_t,\xi_t)\\\Delta_F(t,X_t,\xi_t,1)\\\vdots\\\Delta_F(t,X_t,\xi_t,\xi_t-1)\\\Delta_F(t,X_t,\xi_t,\xi_t+1)\\\vdots\\\Delta_F(t,X_t,\xi_t,M)\ecm,\]
for $0\le t\le T\!\wedge\! T_1 $. 
In other words, the $T$-futures is redundant in this market with $M$ regimes.

\subsection{Utility Maximization}
We consider a portfolio of $M$ futures with different maturities $T_1<T_2<\ldots<T_M$.  We assume that these futures are not redundant, and the interest rate is zero.  We fix the trading horizon $\tilde{T}$, which must not exceed the shortest maturities of the futures in the portfolio. Hence, we require that $\tilde{T} \le T_1$. The dynamic portfolio   contains $\pi^{(i)}_t$ units of the futures $F^{(i)}$ at time $t\le \tilde{T}$. For any strategy $\bm{\pi_t}=(\pi^{(1)}_t,\pi^{(2)}_t,\cdots,\pi^{(M)}_t)'$,
 the wealth process is given by 
\begin{align}
dW^\pi_t &= \sum_{k=1}^{M}\pi^{(k)}_t\,dF_t^{(k)}.\label{WealthProcess}
\end{align}

We now reorganize terms in wealth process \eqref{WealthProcess}. To that end, we define the transformed strategies by
\begin{align}
\tilde{\pi}^{(0)}_t &= \sum_{k=1}^M \pi^{(k)}_t\eta^{(k)}(t,X_t,\xi_t),\label{tpi0}\\
\tilde{\pi}^{(j)}_t &= \sum_{k=1}^M \pi^{(k)}_t\Delta^{(k)}_F(t,X_t,\xi_t,j), \quad \text{for} \quad  j\in E.\label{tpii}
\end{align}
Note that $\tilde{\pi}^{(\xi_t)}_t= \sum_{k=1}^M \pi^{(k)}_t\Delta^{(k)}_F(t,X_t,\xi_t,\xi_t)=0$  since $\Delta^{(k)}_F(t,x,i,i)=0$, for $\forall k =1,\ldots,M$ and $\forall i\in E$, according to \eqref{Difference_Function_N_Regimes}.
In matrix form,  we have
\begin{equation}\label{MatrixFormStrategy}
\bm{\tilde{\pi}_t}=\bm{\Gamma}(t,X_t,\xi_t)\bm{\pi_t},
\end{equation} 
where $
\bm{\tilde{\pi}_t}=(\tilde{\pi}^{(0)}_t,\tilde{\pi}^{(1)}_t,\cdots,\tilde{\pi}^{(\xi_t-1)}_t, \tilde{\pi}^{(\xi_t+1)}_t,\cdots,\tilde{\pi}^{(M)}_t)'.
$  Applying \eqref{FutureSDE}, \eqref{tpi0} and \eqref{tpii}  to \eqref{WealthProcess}, the wealth process becomes
\begin{equation}\label{WealthProcess_New_Strategy}
\begin{aligned}
dW^{\pi}_t &= \bigg(\zeta(\xi_t)\tilde{\pi}^{(0)}_t+\sum_{j\in E\setminus\{\xi_t\}}(q(\xi_t,j)-\tilde{q}(\xi_t,j))\tilde{\pi}^{(j)}_t\bigg)dt\\
&\quad\quad+ \tilde{\pi}^{(0)}_tdZ_t^\P+\int_\R \sum_{j\in E\setminus\{\xi_t\}}\tilde{\pi}^{(j)}_tI_{\{z\in\Delta(\xi_t,j)\}}M^\P(dt,dz) .
\end{aligned}
\end{equation}
Notice that the wealth process is subject to jumps when the market regime switches states. 

We consider a utility maximization approach to determine the optimal futures trading strategy. The investor seeks to maximize the expected utility by dynamically trading the futures continuously over time.  We assume the associated coefficient matrix $\bm{\Gamma}(t,X_t,\xi_t)$ be invertible, and it would act as a bijection mapping  between   $\bm{\pi_t}$ and   $\bm{\tilde{\pi}_t}$. This allows us to solve the portfolio optimization problem by maximizing over $\bm{\tilde{\pi}_t}$. Notice the wealth process \eqref{WealthProcess_New_Strategy} does not explicitly depend on  $X_t$.   The investor solves the following utility maximization problem
 \begin{equation}\label{valuefunction}
u_i(t,w)=\sup_{\tilde{\pi}}\E\big[U(W^{\pi}_{\tilde{T}})|W_t=w,\xi_t=i\big],
\end{equation}
where $U(w)=-e^{-\gamma w}$ is exponential utility function with a constant risk aversion parameter $\gamma >0$.

The investor's value function $u_i(t,w)$ is determined from a system of   Hamilton-Jacobi-Bellman (HJB) equations. Precisely, we have
\begin{equation}
\begin{aligned}
&\partial_t u_i+\max_{\bm{\tilde{\pi}_t}}\bigg\{ \bigg(\zeta_i\tilde\pi^{(0)}_t-\sum_{j\in E\setminus\{i\}}\tilde{q}_{ij}\tilde{\pi}^{(j)}_t\bigg)\partial_w u_i+\frac{(\tilde{\pi}^{(0)}_t)^2}{2}\partial_{ww}u_i\\
&+\sum_{j\in E\setminus\{i\}}q_{ij}\bigg(u_j(t,w+\tilde{\pi}^{(j)}_t)-u_i(t,w)\bigg)\bigg\}=0,
\end{aligned}\label{HJB}
\end{equation}
for $i\in E$ and $t\in[0,\tilde T)$. The terminal condition is $u_i(\tilde T,w)=-e^{-\gamma w}$, for $i\in E$. We have used the shorthand notations: $\tilde{q}_{ij} \equiv \tilde{q}(i,j)$, $q_{ij}\equiv q(i,j)$ and $\zeta_i\equiv\zeta(i)$. 

Performing the optimization in \eqref{HJB} and assuming that $\partial_{ww}u_i\le 0$ (which will be verified later), we obtain the first-order conditions for the optimal strategy:
\begin{equation}
\left\{
\begin{aligned}
&\tilde{\pi}^{(0)*}(t,w,i)=-\zeta_i\frac{\partial_w u_i(t,w)}{\partial_{ww}u_i(t,w)},\\
&\partial_wu_j(t,w+\tilde{\pi}^{(j)*}(t,w,i)) = \frac{\tilde{q}_{ij}}{q_{ij}}\partial_w u_i(t,w),
\end{aligned}
\right.\label{eqnpi}
\end{equation}
for $i\in E$ and $j\in E\setminus\{i\}$. Plugging    this  into   (\ref{HJB}), the HJB equations become 
\begin{equation}\label{HJB12}
\begin{aligned} 
 &\partial_t u_i-\frac{(\zeta_i\partial_w u_i)^2}{2\partial_{ww}u_i}-\sum_{j\in E\setminus\{i\}}\tilde{q}_{ij}\tilde{\pi}^{(j)*}_i\partial_w u_i+\sum_{j\in E\setminus\{i\}}q_{ij}\bigg(u_j(t,w+\tilde{\pi}^{(j)*}_i)-u_i(t,w)\bigg)=0,
 \end{aligned}
 \end{equation}
 for $i\in E$, where we have denoted $\tilde{\pi}^{(j)*}_i\equiv \tilde{\pi}^{(j)*}(t,w,i)$.

We now  consider the transformation for the value function 
\begin{equation}
u_i(t,w)=-e^{-\gamma w+\varphi_i(t)}.\label{ValueFunction}
\end{equation}
Applying this to the first-order conditions \eqref{eqnpi} and HJB equation \eqref{HJB12},  the optimal strategy can be written explicitly 
\begin{equation}
\left\{
\begin{aligned}
&\tilde{\pi}^{(0)*}_i=\frac{\zeta_i}{\gamma },\\
&\tilde{\pi}^{(j)*}_i(t)=-\frac{1}{\gamma}\left(\log\frac{\tilde{q}_{ij}}{q_{ij}}+\varphi_i(t)-\varphi_j(t)\right),
\end{aligned}
\right.\label{OptimalStrategy}
\end{equation}
for $i\in E$ and $j\in E\setminus\{i\}$. We note  that both $\tilde{\pi}^{(0)*}_i$ and $\tilde{\pi}^{(j)*}_i(t)$ do not depend on   wealth $w$. Substituting the transformation \eqref{ValueFunction} and optimal strategy \eqref{OptimalStrategy} into   \eqref{HJB12}, we obtain a  system of ODEs for $\varphi_i(t)$:
\begin{equation}
\begin{aligned}
&\varphi_i'(t)-\sum_{j\in E\setminus\{i\}}\tilde{q}_{ij}\bigg(\varphi_i(t)-\varphi_j(t)\bigg)-\alpha_i=0,\label{eqnv}
\end{aligned}
\end{equation}
where
\begin{equation}\label{alpha}
\alpha_i = \frac{\zeta_i^2}{2}+\sum_{j\in E\setminus\{i\}}\tilde{q}_{ij}\log\frac{\tilde{q}_{ij}}{q_{ij}}-\tilde{q}_{ij}+q_{ij},
\end{equation}
for   $i\in E$ and $t\in[0,\tilde T)$. The  terminal condition is $\varphi_i(\tilde T)=0$ for $i\in E$. This ODE system admits the solution
\begin{equation}\label{phi_solution}
\bm{\varphi}(t) = -\int_t^{\tilde{T}}\exp\left(\bm{\tilde{Q}}(\tilde{T}-s)\right)\bm{\alpha}\,ds\,,
\end{equation}
where $\bm{\tilde{Q}}$ is the generator matrix for $\xi_t$ under measure $\Q$, $
\bm{\varphi}(t) = (\varphi_1(t),\cdots , \varphi_M(t) )'$ and $
\bm{\alpha}=( \alpha_1(t),\cdots , \alpha_M(t) )'$.
With this solution, along with transformation \eqref{ValueFunction},   direction calculations now show that the second-order condition   $\partial_{ww}u_i=\gamma^2 u_i\le 0$ is satisfied. Therefore, the solution to the HJB system (\ref{HJB}) is indeed given by  \eqref{ValueFunction}. In addition, the optimal strategy $\bm{\pi}^*$ can be recovered from  \eqref{MatrixFormStrategy}.

The ODE system \eqref{eqnv} implies a probabilistic representation for $\varphi_i(t)$, given by
\begin{align}\label{viprorep}
 \varphi_i(t)=\E^\Q\bigg[-\int_t^{\tilde T} \alpha(\xi_s)ds \bigg|\xi_t=i\bigg],
\end{align}
where $\alpha(i)=\alpha_i$ is given by \eqref{alpha}. Given  that $x\log x-x+1\ge 0,\,\forall x>0$, it follows that $\alpha(i)\ge 0$ and $\varphi_i(t)\le 0$. Intuitively, $\exp(\varphi_i(t))$ acts like a discounting factor that depends directly on the regime-switching market price of risk.  Since the exponential utility is negative and $\varphi_i(t)\le 0$, this implies that $u_i(t,w)\ge -e^{-\gamma w}$, which means that the investor achieves a higher expected utility by dynamically trading the futures, as compared to only holding the same constant cash amount $w$.

Next, by putting  the optimal strategy $\tilde{\pi}^{(0)*}_i$ and $\tilde{\pi}^{(j)*}_i(t)$ in \eqref{WealthProcess_New_Strategy}, we obtain  the optimal wealth SDE
\begin{equation}
 \begin{aligned}
 dW_t^* &=\frac{1}{\gamma}\bigg(\zeta^2(\xi_t) -\sum_{j\in E\setminus\{\xi_t\}}\bigg(q(\xi_t,j)-\tilde{q}(\xi_t,j)\bigg)\bigg(\log\frac{\tilde{q}(\xi_t,j)}{q(\xi_t,j)}+\varphi(t,\xi_t)-\varphi(t,j)\bigg)\bigg)dt\\
&+\frac{\zeta(\xi_t)}{\gamma}dZ_t^\P-\frac{1}{\gamma}\int_\R\sum_{j\in E\setminus\{\xi_t\}}\bigg(\log\frac{\tilde{q}(\xi_t,j)}{q(\xi_t,j)}+\varphi(t,\xi_t)-\varphi(t,j)\bigg)I_{\{z\in\Delta(\xi_t,j)\}}M^\P(dt,dz),
\end{aligned}
\small \label{WstarSDE}
\end{equation}
where we have denoted $\varphi(t,i)=\varphi_i(t)$, for $i\in E$.
The coefficients of the optimal wealth process in \eqref{WstarSDE} do not  depend on the spot  price $X_t$ but are modulated by the continuous-time Markov chain $\xi_t$. The optimal wealth is connected with the risk-neutral pricing measure and futures prices through the market price of risk $\zeta(\xi_t)$. The investor's risk aversion parameter $\gamma$ also affects the optimal wealth. In particular, a higher $\gamma$ will reduce the magnitude of the drift and volatility of the optimal wealth process.

\begin{remark}
When $M=2$, the ODE system \eqref{eqnv} admits an explicit solution
\begin{equation}\label{PHI_M2}
\begin{aligned}
&\varphi_i(t) = -\frac{1}{\tilde{\lambda}_1+\tilde{\lambda}_2}\bigg((\tilde{\lambda}_2\alpha_1+\tilde{\lambda}_1\alpha_2)(\tilde T-t)+\tilde{\lambda}_i(\alpha_i-\alpha_j)\frac{1-e^{-(\tilde{\lambda}_1+\tilde{\lambda}_2)(\tilde T-t)}}{\tilde{\lambda}_1+\tilde{\lambda}_2}\bigg),
\end{aligned}
\end{equation}
where  $\tilde{\lambda}_i=\tilde{q}_{ij}$, for $(i,j)\in\{(1,2),(2,1)\}$ and    $t\in[0,\tilde T]$. In turn, the optimal strategy is given    by
\begin{equation}\label{OptimalStrategyM2}
\begin{aligned}
&\bcm \pi^{(1)*}_i(t,x) \\ \pi^{(2)*}_i(t,x) \ecm \\
=&-\frac{1}{\gamma b(t,x,i)((F^{(2)}_j(t,x)-F^{(2)}_i(t,x))\partial_xF^{(1)}_i(t,x)-(F^{(1)}_j(t,x)-F^{(1)}_i(t,x))\partial_xF^{(2)}_i(t,x))}\\
&\bcm-\zeta_i(F_j^{(2)}(t,x)-F_i^{(2)}(t,x))-b(t,x,i)\partial_xF_j^{(2)}(t,x)( \log\frac{\tilde{\lambda}_i}{\lambda_i}+\varphi_i(t)-\varphi_j(t))\\\zeta_i(F_j^{(1)}(t,x)-F_i^{(1)}(t,x))+b(t,x,i)\partial_xF_j^{(1)}(t,x)( \log\frac{\tilde{\lambda}_i}{\lambda_i}+\varphi_i(t)-\varphi_j(t))\ecm,
\end{aligned}
\end{equation}
where    $\lambda_i=q_{ij}$, for    $(i,j)\in\{(1,2),(2,1)\}$ and  $(t,x)\in[0,\tilde T]\times \R$. 
\end{remark}

\subsection{Certainty Equivalent}
In order to quantify   the value of trading futures to the investor, we define the investor's certainty equivalent  associated with the utility maximization problem. The certainty equivalent is the guaranteed cash amount   that would yield the same  utility as that from dynamically trading futures according to \eqref{valuefunction}. This amounts to applying the inverse of the utility function to the value function in (\ref{ValueFunction}), that is,
\begin{align}
c_i(t,w) := U^{-1}(u_i(t,w))= w-\frac{\varphi_i(t)}{\gamma}. \label{ce}
\end{align}
Therefore, the certainty equivalent is the sum of the investor's wealth $w$ and a time-dependent component  $-\frac{\varphi_i(t)}{\gamma}$.  From the probabilistic representation  in (\ref{viprorep}), we know that $\varphi_i(t)$ is  negative. The certainty equivalent is also inversely proportional to the risk aversion parameter $\gamma$, which means  that a more risk averse investor has a lower certainty equivalent, valuing the futures trading opportunity less. 

Under our regime-switching framework, the certainty equivalent is the same across different underlying models as long as the market prices of risk stay  the same. Therefore, without picking a specific model, we can still compute the certainty equivalent.   Figure \ref{CertaintyEquivalent} illustrates the certainty equivalents under two risk aversion levels under two regimes in a two-regime market.  Under each regime, the  less risk averse investor ($\gamma=0.5$) has a higher certainty equivalent than the more  risk averse investor ($\gamma=2$). According to the parameters for the two regimes, Regime 2 has a higher market price of risk than regime 1, which explains that the investor's certainty equivalent is higher in regime 2  than in regime 1.  All else being equal, the certainty equivalent is higher when there is more time to trade. Hence, as the trading horizon reduces to zero, the certainty equivalent converges to the initial wealth $w$, which is set to be 1 in this example. This means that the second term in \eqref{ce} converges to zero since $\varphi_i(t)\rightarrow 0$ as $t\rightarrow \tilde{T}$.

\begin{figure}[h]
\centering
\includegraphics[trim=0.5cm   0.4cm  1cm  1cm,clip,width=3.4in]{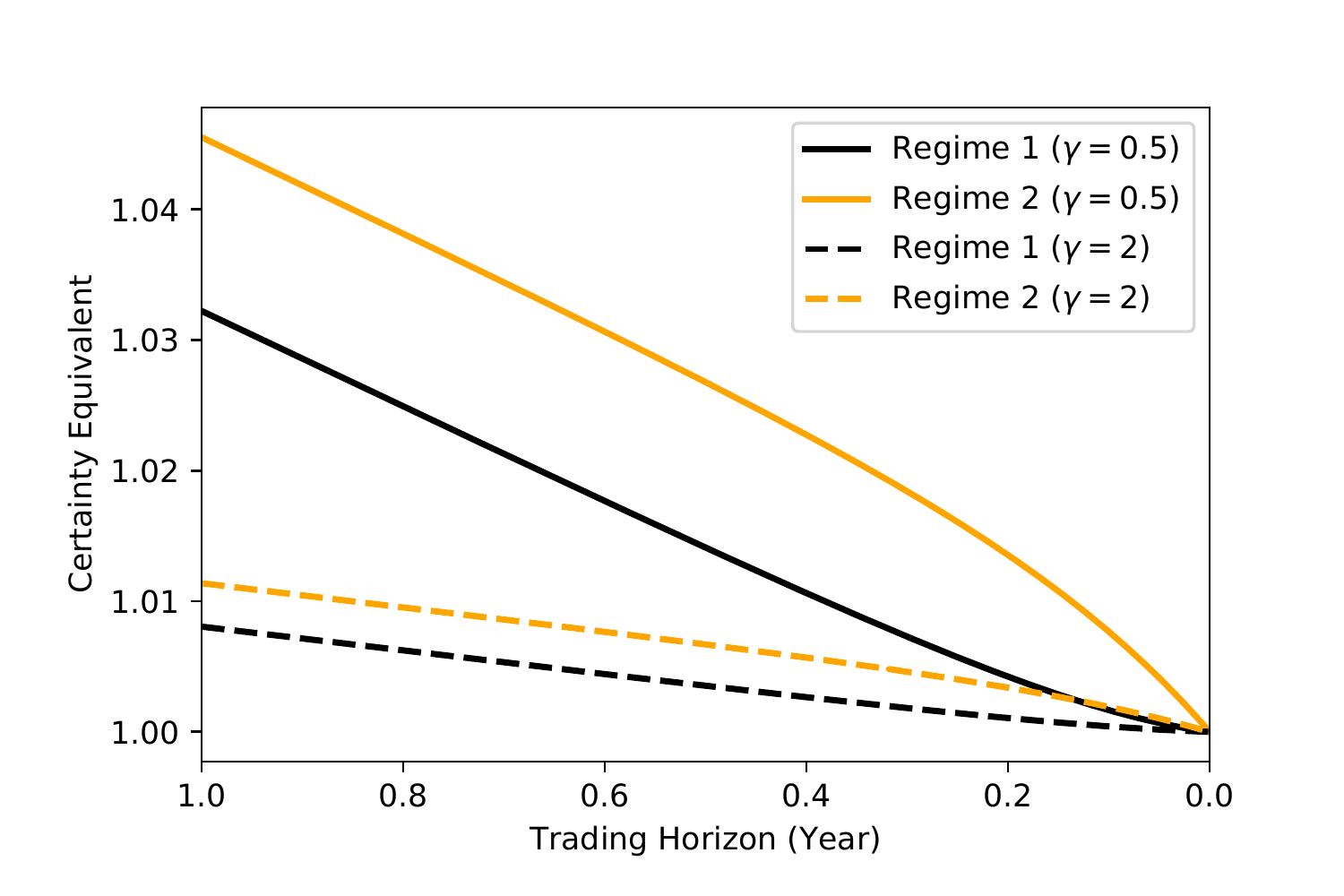}
\caption{The certainty equivalents corresponding to two different risk aversion levels  in a two-regime market. They are plotted as functions of the trading horizon, for  $\gamma=0.5$ (solid lines) and $\gamma = 2$ (dashed lines) in regime 1 (dark color) and regime 2 (light color). Common parameters are  $\tilde{q}_{12}=q_{12} = 0.8$, $\tilde{q}_{21}=q_{21}=0.6$, $w=1$, $\zeta_1=0.1$ and $\zeta_2=0.3$.}\label{CertaintyEquivalent}
\end{figure}
 
 In the next two sections, we consider two model specifications and provide numerical examples.

\newpage
 \section{Regime-Switching Geometric Brownian Motion}\label{sect-GBM}
Suppose the   log-price of the underlying asset follows the SDE
\begin{equation}
dX_t = \mu(\xi_t) dt+\sigma(\xi_t)dZ_t^\Q,
\end{equation} under the risk-neutral  measure $\Q$. We call this model the Regime-Switching Geometric Brownian Motion (RS-GBM) because without $\xi$ the spot price $S$ is simply a GBM. This model belongs to our regime-switching framework discussed in the previous section. Indeed, this amounts to setting the coefficients  in SDE (\ref{log-price}) to be
\begin{align}\tilde a(t,X_t,\xi_t)= \mu(\xi_t), \quad \text{ and } \quad  b(t,X_t,\xi_t) =\sigma(\xi_t).\label{ab}\end{align}

Substituting \eqref{ab} into  (\ref{FuturePDE}), we  obtain the PDE system for the futures price function under this model. Precisely, for $i\in E$, we have
\begin{align}
\partial_t F_i+\mu_i\partial_x F_i+\frac{\sigma_i^2}{2}\partial_{xx} F_i+\sum_{j\in E\setminus \{i\}}\tilde {q}(i,j)( F_j- F_i)=0,
\end{align}
with $\mu_i=\mu(i)$ and $\sigma_i=\sigma(i)$. The terminal condition is $F_i(T,x)=e^x$, $x\in\R$.

Under this model, the  futures price admits the separation of variables:
\begin{equation}
F_i(t,x) =  e^x\,g_i(t),\label{FutureGBM}
\end{equation}
where  $(g_i(t))_{i=1,\ldots, M}$ solve the system of  ODEs:
\begin{align}
 g'_i(t)+(\mu_i+\frac{\sigma_i^2}{2})g_i(t)+\sum_{j\in E\setminus \{i\}} \tilde{q}(i,j)(g_j(t)-g_i(t))=0,\label{eqng}
\end{align}
for $t\in[0,T)$, with the terminal condition $g_i(T)=1$, for $ i=1,\ldots,M$. Defining $\bm{g}(t) = (g_1(t), \ldots, g_M(t))'$,   we can write the solution as
\begin{equation}
\bm{g}(t)= \exp\bigg((\bm{G}+\bm{\tilde Q})(T-t)\bigg)\textbf{1},\label{g_GBM_solution}
\end{equation}
where  $\bm{\tilde Q}$ is the generator matrix under the measure $\Q$ and 
\begin{equation}
\bm{G}=\mbox{diag}(\frac{2\mu_1+\sigma^2_1}{2},\frac{2\mu_2+\sigma^2_2}{2},\ldots,\frac{2\mu_M+\sigma^2_M}{2}).
\end{equation}
In addition, the ODE system \eqref{eqng} implies a probabilistic representation for $g_i(t)$:
\begin{align}\label{g_presentation}
g_i(t)=\E^\Q\bigg[\exp\left(\int_t^{ T} \mu(\xi_s)+\frac{\sigma^2(\xi_s)}{2}ds\right) \bigg|\xi_t=i\bigg].
\end{align}
To sum up, the futures price in regime $i$ is given by 
\begin{equation}
F_i(t,x)=\exp(x)\bigg(\exp\bigg((\bm{G}+\bm{\tilde Q})(T-t)\bigg)\textbf{1}\bigg)_{i},
\label{FuturesGBMSolution}
\end{equation}
where the subscript $i$ denotes the $i$th entry of the vector.


In turn, using \eqref{FutureSDE}, the futures price satisfies the SDE  
\begin{equation}\label{FutureSDEGBM}
\begin{aligned}
dF_t&=\bigg(\sigma(\xi_t)F(t,X_t,\xi_t)\zeta(\xi_t)+\!\!\!\sum_{j\in E\setminus\{\xi_t\}}\!\!\!\!\!(q(\xi_t,j)-\tilde{q}(\xi_t,j))\bigg(F(t,X_t,j)-F(t,X_t,\xi_t)\bigg) \bigg)dt\\
&+\sigma(\xi_t)F(t,X_t,\xi_t)dZ_t^{\P}+\int_\R \sum_{j\in E\setminus\{\xi_t\}}\bigg(F(t,X_t,j)-F(t,X_t,\xi_t)\bigg)I_{\{z\in\Delta(\xi_t,j)\}}M^\P(dt,dz).
\end{aligned}
\end{equation}

Next, we consider a portfolio of futures with $M$ different maturities $T_1<T_2<\cdots<T_M$. If the associated coefficient matrix $\bm{\Gamma}$ is invertible, we can apply results in the Section \ref{sect-utility}. To that end, we have following proposition for coefficient matrix $\bm{\Gamma}$.\\

\begin{proposition}\label{Prop-GBM-Gamma}
If $\bm{\Gamma}(t_0,x,i)$ is invertible for some specific time $t_0\le T_1$, then $\bm{\Gamma}(t,x,i)$ is invertible for any time $t\le T_1$.
\end{proposition}
\begin{proof}
 By \eqref{coefficient_Definition}, the coefficient matrix is given by
\begin{equation}
\bm{\Gamma}(t,x,i)=e^x\bcm \sigma_ig^{(1)}_i(t)&\sigma_ig^{(2)}_i(t)&\cdots& \sigma_ig^{(M)}_i(t)\\
 g^{(1)}_1(t)-g^{(1)}_i(t)&g^{(2)}_1(t)-g^{(2)}_i(t)&\cdots&g^{(M)}_1(t)-g^{(M)}_i(t)\\
 \vdots& \vdots&\vdots &\vdots\\
  g^{(1)}_{i-1}(t)-g^{(1)}_i(t)&g^{(2)}_{i-1}(t)-g^{(2)}_i(t)&\cdots&g^{(M)}_{i-1}(t)-g^{(M)}_i(t)\\
  g^{(1)}_{i+1}(t)-g^{(1)}_i(t)&g^{(2)}_{i+1}(t)-g^{(2)}_i(t)&\cdots&g^{(M)}_{i+1}(t)-g^{(M)}_i(t)\\
  \vdots& \vdots&\ddots &\vdots\\
  g^{(1)}_M(t)-g^{(1)}_i(t)&g^{(2)}_M(t)-g^{(2)}_i(t)&\cdots&g^{(M)}_M(t)-g^{(M)}_i(t)\ecm,
\end{equation}
for $i\in E$.
The determinant of coefficient matrix $\bm{\Gamma}$ is denoted by 
\begin{align}
\Phi_i(t,x)&=\mbox{det }\bm{\Gamma}(t,x,i),  \qquad i\in E.
\end{align}

We define the matrix $\bm{\tilde{\Gamma}}$ by adding $1/\sigma_i$ of the first row to other rows in matrix $\bm{\Gamma}$. Then, we define the matrix $\bm{H_k}$ by taking $t$-derivative in the $k$-th row of matrix $\bm{\tilde{\Gamma}}$. For example,
\begin{equation}
\bm{H_1}(t,x,i)=e^x\bcm \sigma_i\frac{d}{dt}g^{(1)}_i(t)&\sigma_i\frac{d}{dt}g^{(2)}_i(t)&\cdots& \sigma_i\frac{d}{dt}g^{(M)}_i(t)\\
 g^{(1)}_1(t)&g^{(2)}_1(t)&\cdots&g^{(M)}_1(t)\\
 \vdots& \vdots&\vdots &\vdots\\
  g^{(1)}_{i-1}(t)&g^{(2)}_{i-1}(t)&\cdots&g^{(M)}_{i-1}(t)\\
  g^{(1)}_{i+1}(t)&g^{(2)}_{i+1}(t)&\cdots&g^{(M)}_{i+1}(t)\\
  \vdots& \vdots&\ddots &\vdots\\
  g^{(1)}_M(t)&g^{(2)}_M(t)&\cdots&g^{(M)}_M(t)\ecm.
\end{equation}
Then, holding the log-price $x$ fixed, we differentiate to get
\begin{equation}\label{Phi_H_k}
\frac{d}{dt}\Phi_i(t,x) = \frac{d}{dt} \det \bm{ \tilde{\Gamma}}(t,x,i) = \sum_{k=1}^M \det \bm{H_k}(t,x,i).
\end{equation}
Applying \eqref{eqng} and  $\sum_{j\in E\setminus\{i\}}\tilde{q}(i,j)=-\tilde{q}(i,i)$, we get
\begin{equation}\label{H_k_Phi}
\sum_{k=1}^M\det \bm{H_k}(t,x,i) =\sum_{k=1}^M-\left(\mu_k+\frac{\sigma_k^2}{2}+\tilde{q}(k,k)\right)\Phi_i(t,x),
\end{equation}
for $i \in E$.

Combining \eqref{Phi_H_k} and \eqref{H_k_Phi},   we have
\begin{equation}
\frac{d}{dt}\Phi_i(t,x)+ \sum_{k=1}^M\bigg(\mu_k+\frac{\sigma_k^2}{2}+\tilde{q}(k,k)\bigg)\Phi_i(t,x)= 0.
\end{equation}
Then, for any $t_0,t\le T_1$,  $\Phi_i(t,x)$ satisfies
\begin{equation}\label{Phi_Relation}
\Phi_i(t,x) = \exp\bigg(-\sum_{k=1}^M\bigg(\mu_k+\frac{\sigma_k^2}{2}+\tilde{q}(k,k)\bigg)(t-t_0)\bigg)\Phi_i(t_0,x). 
\end{equation}
Thus, if $\bm{\Gamma}(t_0,x,i)$ is invertible for some specific time $t_0\le T_1$, then $\bm{\Gamma}(t,x,i)$ is invertible for any time $t\le T_1$.
\end{proof}\\

\begin{example}
Assume a market with two regimes, i.e. $E=\{1,2\}$. The coefficient matrix $\bm{\Gamma}$ for futures pair $(F^{(1)},F^{(2)})$ is given explicitly by
\begin{equation}
\bm{\Gamma}(t,x,i)=e^x\bcm \sigma_ig^{(1)}_i(t)&\sigma_ig^{(2)}_i(t)\\
 g^{(1)}_j(t)-g^{(1)}_i(t)&g^{(2)}_j(t)-g^{(2)}_i(t)
 \ecm,
\end{equation}
for $(i,j)\in\{(1,2),(2,1)\}$. Then, the matrix determinant $\Phi_i(t,x)$ is also explicit:
\begin{equation}
\Phi_i(t,x) = e^{2x}\sigma_i(g^{(1)}_i(t)g^{(2)}_j(t)-g^{(1)}_j(t)g^{(2)}_i(t)),
\end{equation}
for $(i,j)\in\{(1,2),(2,1)\}$. 

Applying Proposition \ref{Prop-GBM-Gamma}, the coefficient matrix $\bm{\Gamma}$ is invertible for any time $t\le T_2$, if and only if $\Phi_i(T_2,x)=e^{2x}\sigma_i(g^{(2)}_j(T_2)-g^{(2)}_i(T_2))\neq 0$, which is equivalent to  $\mu_1+\sigma_1^2/2 \neq \mu_2+\sigma_2^2/2$ accorinding to the probabilistic representation \eqref{g_presentation}.

If $\mu_1+\sigma_1^2/2 \neq \mu_2+\sigma_2^2/2$, we can apply the results in Section \ref{sect-utility}. The value function $u_i(t,w)$ satisfies
\begin{equation}
u_i(t,w)=-e^{-\gamma w+\varphi_i(t)},
\end{equation}
where $\varphi_i(t)$ is given by \eqref{PHI_M2}. Applying \eqref{OptimalStrategyM2}, we immediately obtain the optimal strategy 
\begin{equation}
\begin{aligned}
\bcm \pi^{(1)*}_i(t,x) \\ \pi^{(2)*}_i(t,x) \ecm =&-\frac{e^{-x}}{\gamma\sigma_i(g_i^{(1)}(t)g_j^{(2)}(t)-g_j^{(1)}(t)g_i^{(2)}(t))} \\
&\bcm-\zeta_i(g_j^{(2)}(t)-g_i^{(2)}(t))-\sigma_ig_i^{(2)}(t)\left(\log\frac{\tilde{\lambda}_i}{\lambda_i}+\varphi_i(t)-\varphi_j(t)\right)\\ \zeta_i(g_j^{(1)}(t)-g_i^{(1)}(t))+\sigma_ig_i^{(1)}(t)\left(\log\frac{\tilde{\lambda}_i}{\lambda_i}+\varphi_i(t)-\varphi_j(t)\right)\ecm,
\end{aligned}
\label{OSGBM}
\end{equation}
for  $(i,j)\in\{(1,2),(2,1)\}$ and $(t,x)\in[0,\tilde T]\times \R$.

If $\mu_1+\sigma_1^2/2 = \mu_2+\sigma_2^2/2$, the futures prices will be the same in two states according to  \eqref{g_presentation}. Therefore, the futures price SDE  becomes 
\begin{equation}
 \begin{aligned}
&dF_t=\sigma(\xi_t)F(t,X_t,\xi_t)\zeta(\xi_t)dt+\sigma(\xi_t)F(t,X_t,\xi_t)dZ_t^{\P},\end{aligned}
 \end{equation} 
which, in contrast to \eqref{FutureSDEGBM},  has continuous paths. 

 
%

\end{example}

\newpage
\subsection{Numerical Illustration}\label{Numerical_GBM}
We simulate the sample paths for spot price, futures price, optimal investment, and optimal wealth process. The regime switching between two states, with transition probabilities  $q_{12}=2$ and $q_{21}=4$ which are entries of generator matrix $\bm{Q}$.  The trading horizon $\tilde{T}=0.6$ which is no greater than the maturity of the futures contract $T_1=0.6$ and $T_2=0.8$. All parameters   are summarized in the Table \ref{Parameters_Table_GBM}.\\

\begin{table}[h]
\centering
\begin{tabular}{c|c|c|c|c|ccccccccc}
\hline
$\tilde q_{12}=q_{12}$& $\tilde q_{21}=q_{21}$& $\zeta(1)$& $\zeta(2)$&
 $\gamma$& $\tilde T$\\
\hline
$2$&$4$& 0.1 & 0.3  & 1& 0.6\\
\hline
\hline

$\mu_1$& $\mu_2$&$\sigma_1$& $\sigma_2$&
$T_1$ & $T_2$  \\
\hline
-0.2 & 0.2& 0.2& 0.3& 0.6& 0.8 \\
\hline
\end{tabular}
\caption{Parameters for the RS-GBM model for Figures \ref{spot_path_GBM}--\ref{Phi_path_GBM}. }
\label{Parameters_Table_GBM}
\end{table}

As shown in Figure \ref{spot_path_GBM}, the market starts in regime 2, then switches  to regime 1  at time $t_1$ before returning to regime 2 at time $t_2$. Under RS-GBM  model, the futures price (see \eqref{FuturesGBMSolution}) tends to  amplify the   spot price. Thus, the futures price is   volatile relative to the spot price. Moreover, each regime switch can cause an instant jump in the  futures price but not in the spot price.  Since the trading horizon coincides with the maturity of the $T_1$-futures, the corresponding futures price  converges to the spot price towards the end.  \\

In Figure \ref{position_path_GBM}, we illustrate the  sample paths of the optimal positions in  the two futures over the trading horizon. As we can see, the  optimal positions are of opposite signs, meaning that the investor will go long on the $T_1$-Futures and short on the $T_2$-Futures. Long-short strategies are very common in futures trading, and they can help mitigate  the effect of regime switching. As the market switches from regime 2 to regime 1 at time $t_1$, the magnitude of the futures position is  reduced immediately.   The investor then take larger long-short positions when the market returns to regime 2 from regime 1 at time $t_2$.  According to the sample paths, the optimal positions tend to decay in time  during each regime, meaning that the investor gradually reduces investments towards the end of the trading horizon.  \\

Figure \ref{wealth_path_GBM} shows the  wealth process corresponding to the optimal trading strategy in  Figure \ref{position_path_GBM}. The long-short strategy tends to reduce the shock in portfolio vlaue due to  regime switching. Therefore, no  sizable jumps in the wealth process are observed at the regime-swtiching times $t_1$ and $t_2$.  As we can see, the   wealth process decreases slightly  in regime 1. This is intuitive since regime 1 has lower risk premium than in regime 2. The wealth increases sharply in regime 2  towards the end of the trading horizon, even though   the positions are gradually reduced over time.  \\

 In Figure \ref{Phi_path_GBM}, we plot the sample path for the absolute value of coefficient matrix determinant  $|\Phi|$ corresponding to the sample paths of the spot and futures prices.  Recall from   \eqref{Phi_Relation},  $|\Phi|$ is exponentially increasing or decreasing overtime. Moreover, in our numerical settings, the term $\tilde{q}(i,i)=-\tilde{q}(i,j)$ dominates the term $\mu_i+\sigma_i^2/2$. Therefore, $|\Phi|$ appears to be  exponentially increasing in time. Since $|\Phi|$ is not equal to 0, $\Gamma$ is invertible and we are able to apply the results in the Section \ref{sect-utility} for the RS-GBM model.

\begin{figure}
\centering
\includegraphics[trim=0.5cm   0.4cm  1.5cm  1.7cm,clip,width=3.5in]{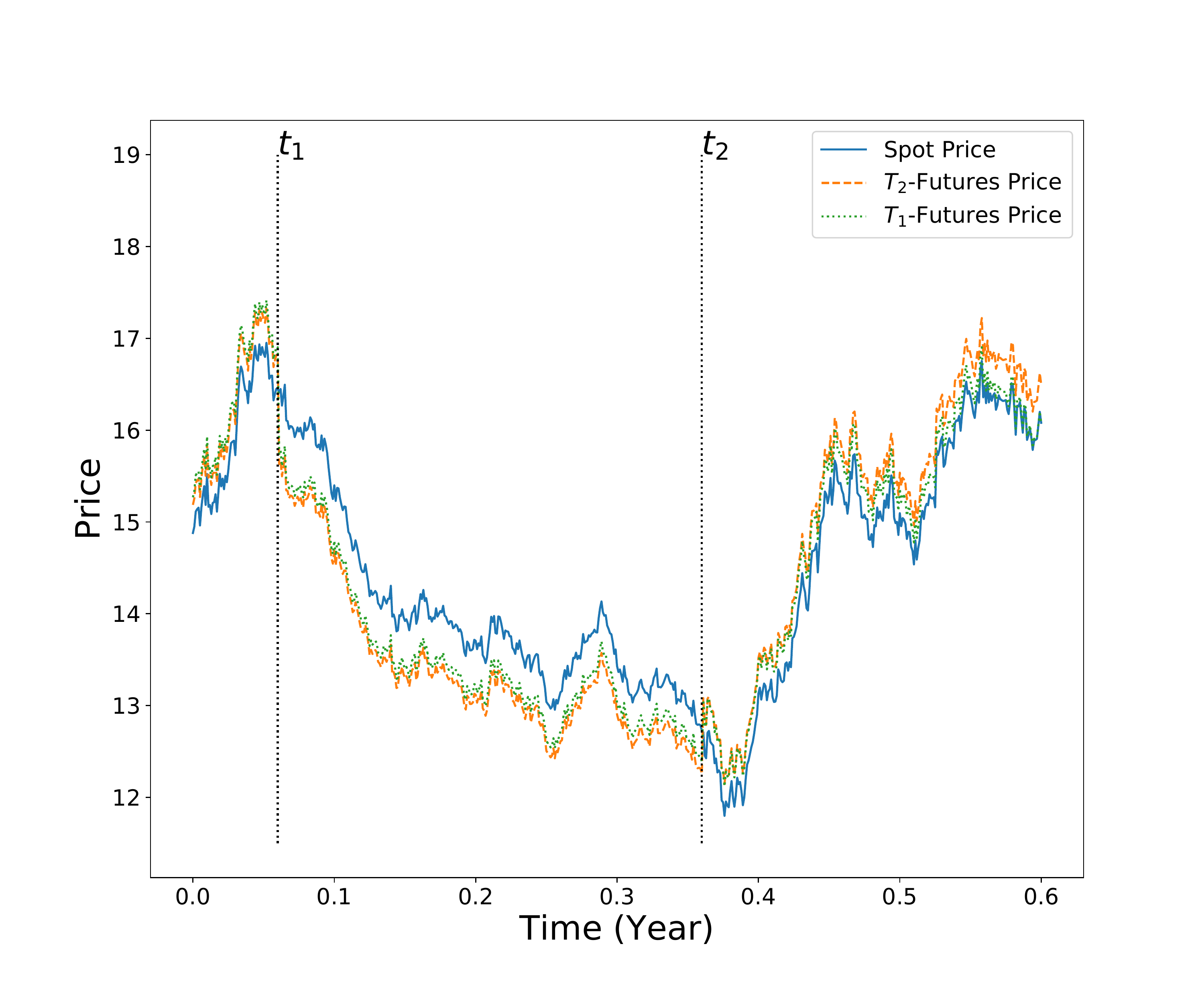}
\caption{Sample paths of the spot price, $T_1$-futures price, and $T_2$-futures price  over the trading horizon under the RS-GBM model. The market starts in regime 2, then switches to regime 1 at time $t_1$, before switching back to regime 2 at time $t_2$.   }\label{spot_path_GBM}
\end{figure}

\begin{figure}
\centering
\includegraphics[trim=0.5cm   0.4cm  1.5cm  1.7cm,clip,width=3.3in]{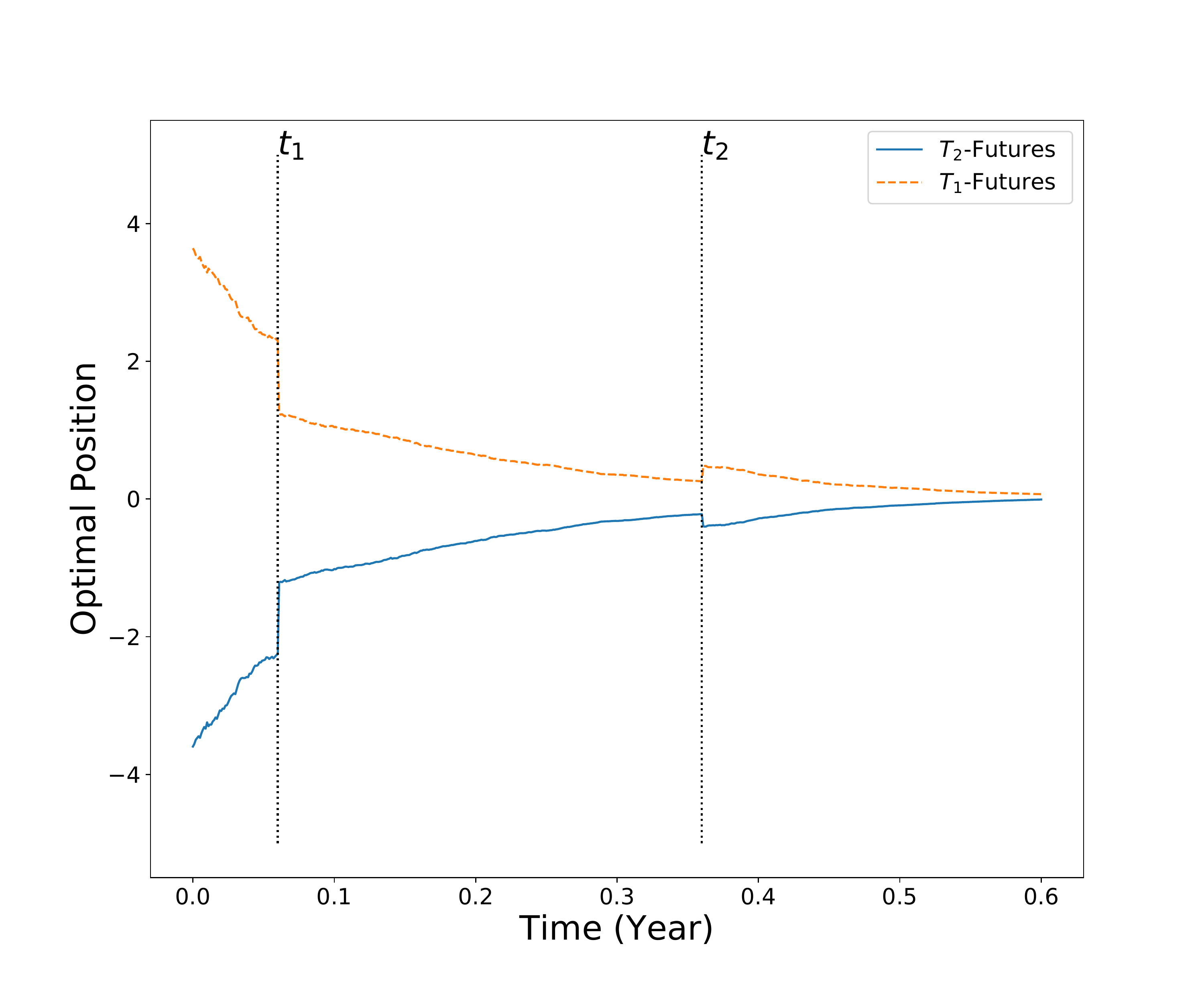}
\caption{Sample paths of the optimal positions ($\pi^{(1)},\pi^{(2)}$) in the  $T_1$-futures and $T_2$-futures respectively  under the RS-GBM model.  The futures positions tend to decrease over time and approaches zero near the end of the trading horizon. Jumps in both  positions occur at the regiem-switching times $t_1$ and $t_2$. }\label{position_path_GBM}
\end{figure}

\clearpage
\begin{figure}
\centering
\includegraphics[trim=0.5cm   0.4cm  1.5cm  1.7cm,clip,width=3.3in]{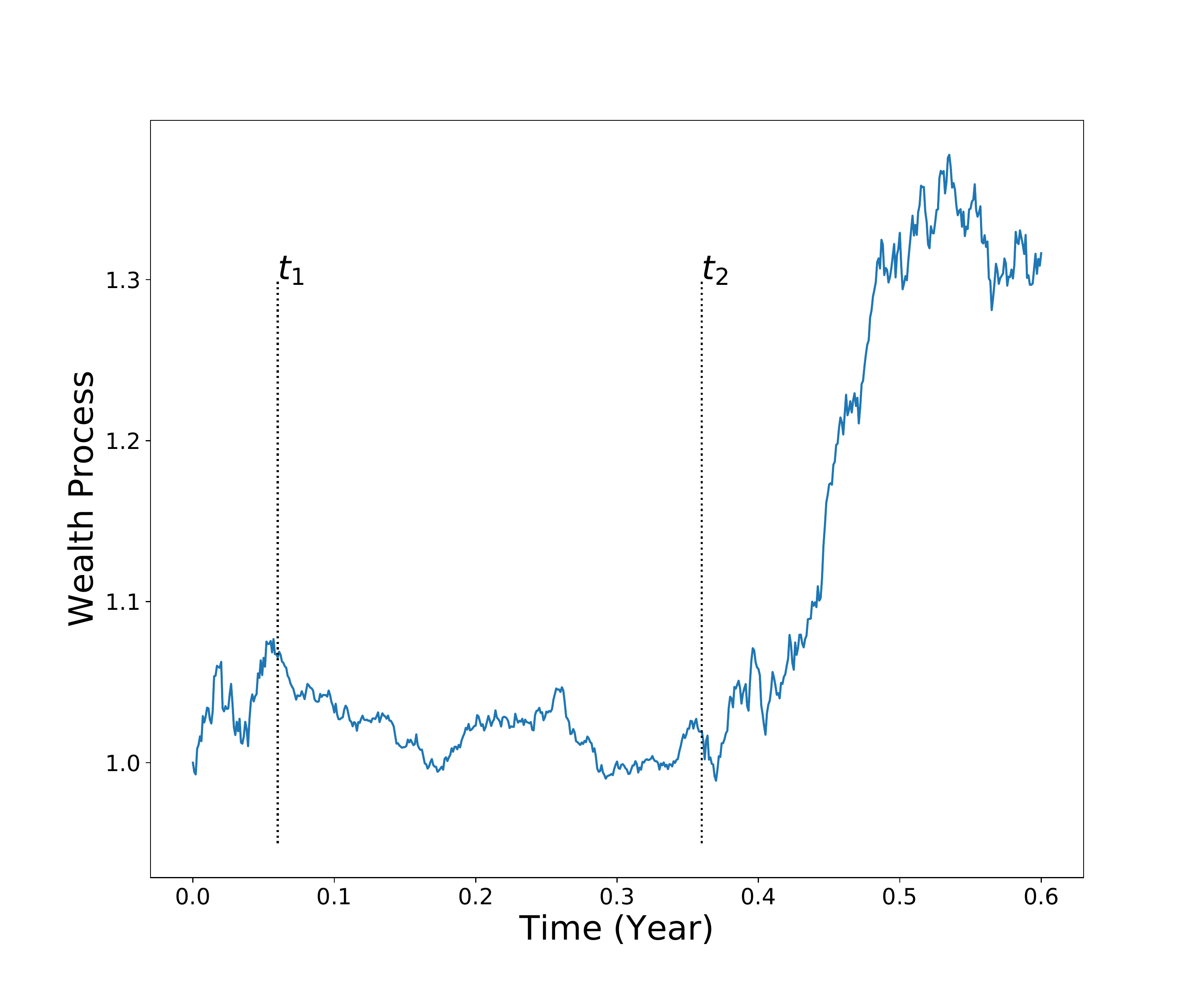}
\caption{Sample path of the optimal wealth process over the trading horizon under the RS-GBM model.}\label{wealth_path_GBM}
\end{figure}




\begin{figure}
\centering
\includegraphics[trim=0.5cm   0.4cm  1.5cm  1.7cm,clip,width=3.3in]{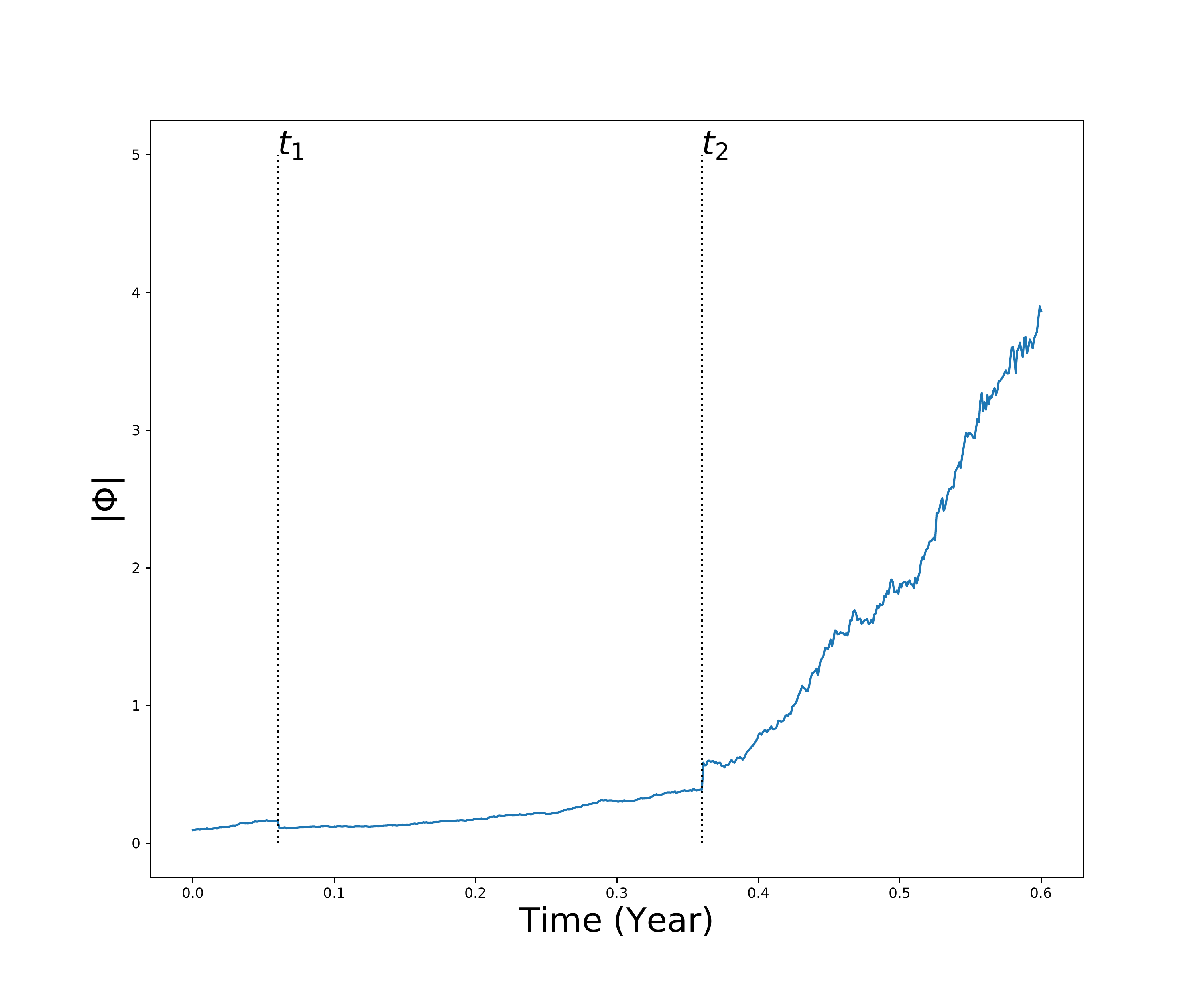}
\caption{Sample path of the absolute value of the determinant $\Phi$ for the coefficient matrix $\Gamma$ under the RS-GBM model.}\label{Phi_path_GBM}
\end{figure}

 \clearpage
\section{Regime-Switching Exponential Ornstein-Uhlenbeck Model}\label{sect-EXP-OU}
  As is well known, the Exponential Ornstein-Uhlenbeck  process and its variations are widely used to model   commodity prices. We now consider  the Regime-Switching Exponential Ornstein-Uhlenbeck (RS-XOU) model and illustrate the optimal trading strategies under this model.  In the  RS-XOU model, the   log spot price   evolves according to
\begin{equation}
dX_t = \kappa(\xi_t)(\theta(\xi_t)-X_t)dt+\sigma(\xi_t)dZ_t^\Q,\label{OU}
\end{equation}
where $\kappa(\xi_t)$, $\theta(\xi_t)$ and $\sigma(\xi_t)$ are the functions of regimes. This amounts to setting $\tilde a(t,X_t,\xi_t)=\kappa(\xi_t)(\theta(\xi_t)-X_t)$ and $b(t,X_t,\xi_t)=\sigma(\xi_t)$ in    (\ref{log-price}).


\subsection{Futures Dynamics and Utility Maximization}
The  futures price function $F_i(t,x)$ satisfies   PDE (\ref{FuturePDE}). Substituting (\ref{OU}) into (\ref{FutureSDE}), we obtain  the $\mathbb{Q}$-dynamics for futures price $F_t$:
\begin{equation}
\begin{aligned}
dF_t &=\eta(t,X_t,\xi_t)dZ_t^\Q+\int_\R \sum_{j\in E\setminus\{\xi_t\}}\bigg(F(t,X_t,j)-F(t,X_t,\xi_t)\bigg)I_{\{z\in\Delta(\xi_t,j)\}}M^\Q(dt,dz),
\end{aligned}
\end{equation}
where the volatility term is given by
\begin{equation}
\eta(t,x,i) = \sigma_i\partial_x F_i(t,x).
\end{equation}
In addition, under the measure $\P$,
\begin{equation}
\begin{aligned}
 dF_t&=\bigg(\zeta(\xi_t)\eta(t,X_t,\xi_t)+\sum_{j\in E\setminus\{\xi_t\}}(q(\xi_t,j)-\tilde{q}(\xi_t,j))\bigg(F(t,X_t,j)-F(t,X_t,\xi_t)\bigg) \bigg)dt\\
&+\eta(t,X_t,\xi_t)dZ_t^{\P}+\int_\R \sum_{j\in E\setminus\{\xi_t\}}\bigg(F(t,X_t,j)-F(t,X_t,\xi_t)\bigg)I_{\{z\in\Delta(\xi_t,j)\}}M^\P(dt,dz).
\end{aligned}
\end{equation}

\begin{example}
We consider a portfolio of futures with two different maturities $T_1$ and $T_2$ in a two-state   market, i.e. $E=\{1,2\}$. The associated coefficient matrix is given by 
\begin{equation}
\bm{\Gamma}(t,x,i)=\bcm \sigma_i\partial_x F^{(1)}_i(t,x)& \sigma_i\partial_x F^{(2)}_i(t,x)\\ F^{(1)}_j(t,x)-F^{(1)}_i(t,x)&F^{(2)}_j(t,x)-F^{(2)}_i(t,x) \ecm,
\end{equation}
for $(i,j)\in\{(1,2),(2,1)\}$.
Moreover, we assume $\bm{\Gamma}(t,x,i)$ be invertible. Under this assumption, we can apply the result in Section \ref{sect-utility}. The value function $u_i(t,w)=-e^{-\gamma w+\varphi_i(t)}$, where $\varphi_i(t)$ is given by \eqref{PHI_M2}, holds. 

Applying  \eqref{OptimalStrategyM2}, we immediately obtain the optimal strategy 
\begin{equation}\label{OptimalStrategyOU}
\begin{aligned}
\bcm \pi^{(1)*}_i(t,x) \\ \pi^{(2)*}_i(t,x) \ecm =-\frac{1}{\gamma\sigma_i((F^{(2)}_j(t,x)-F^{(2)}_i(t,x))\partial_xF^{(1)}_i(t,x)-(F^{(1)}_j(t,x)-F^{(1)}_i(t,x))\partial_xF^{(2)}_i(t,x))} \\
 \bcm-\zeta_i(F_j^{(2)}(t,x)-F_i^{(2)}(t,x))-\sigma_i\partial_xF_j^{(2)}(t,x)( \log\frac{\tilde{\lambda}_i}{\lambda_i}+\varphi_i(t)-\varphi_j(t))\\\zeta_i(F_j^{(1)}(t,x)-F_i^{(1)}(t,x))+\sigma_i\partial_xF_j^{(1)}(t,x)( \log\frac{\tilde{\lambda}_i}{\lambda_i}+\varphi_i(t)-\varphi_j(t))\ecm,
\end{aligned}
\end{equation}
for $(i,j)\in \{(1,2),(2,1)\}$ and  $(t,x)\in[0,\tilde T]\times \R$. 
\end{example}

\begin{remark}\label{remark1}
If the speed of mean reversion is identical in all regimes, i.e. $\kappa_i=\kappa$ $\forall  i$, then the futures price admits the following form
\begin{equation}
F_i(t,x) = \exp\bigg(e^{- \kappa (T-t)}x\bigg)h_i(t),\label{FutureOU}
\end{equation}
where $(h_i(t))_{i=1,\ldots, M}$ solves the ODE system 
\begin{align}
h'_i(t)+ \bigg(\kappa\theta_ie^{-\kappa(T-t)}+\frac{\sigma^2_i}{2}e^{-2\kappa(T-t)}\bigg)h_i(t)+\sum_{j\in E\setminus\{ i\}} \tilde q(i,j)(h_j(t)-h_i(t))=0,\label{eqnh}
\end{align}
with the terminal condition $h_i(T)=1$, for $i=1,\ldots,M$.
 
The price formula \eqref{FutureOU} means that $\eta(t,x,i) = \sigma_i e^{-\kappa (T-t)}F_i(t,x)$ in (\ref{FutureSDE}). As a result,  the futures price satisfies the SDE
\begin{equation}
\begin{aligned}
 dF_t&=\sigma(\xi_t) e^{-\kappa (T-t)}\exp\bigg(e^{- \kappa (T-t)}X_t\bigg)h(t,\xi_t)dZ_t^\Q\\
&+\exp\bigg(e^{- \kappa (T-t)}X_t\bigg)\int_\R \sum_{j\in E\setminus\{\xi_t\}}\bigg(h(t,j)-h(t,\xi_t)\bigg)I_{\{z\in\Delta(\xi_t,j)\}}M^\Q(dt,dz), 
\end{aligned}
\end{equation}
where we have denoted $h(t,i)\equiv h_i(t)$. 
\end{remark}

\subsection{Numerical Method for Futures Price}\label{sect-numerical}
We now discuss the numerical procedue to compute  the futures price under the RS-XOU model in this section. We have computed and checked our prices using the finite difference method (FDM) and fast fourier transform (FFT).
\subsubsection{Finite Difference Method (FDM)}
We first re-write  futures price  PDE (\ref{FuturePDE}) in terms of $s=e^x$   to get
\begin{align}
\partial_t F_i+\bigg(\kappa_i(\theta_i-\ln s)+\frac{\sigma_i^2}{2}\bigg)s\partial_s F_i+\frac{\sigma^2_is^2}{2}\partial_{ss} F_i+\sum_{j\in E\setminus\{ i\}} \tilde q(i,j)( F_j- F_i)=0,\label{eqns}
\end{align}
for $(t,s)\in[0,T)\times\R_+$, with the terminal condition $F_i(T,s)=s$. Then we can apply the Crank-Nicolson method to equation (\ref{eqns}) directly with Dirichlet boundary conditions. The method is standard, so we omit the details here.

\subsubsection{Fourier Space Time-stepping (FST)}
 We begin first by discretizing the continuous-time Markov chain $\xi_t$  with time step of  size $\delta t$, and we  keep $\xi_t$ constant on each time interval  $(t_n,t_{n+1}]$ $(t_n=n\delta t)$, for $n=0,\ldots,T/\delta t-1$, with transition probabilities
\begin{equation}
P_{kl}:=
\left\{
\begin{aligned}
&1+\tilde q_{ll}\delta t, &k=l,\\
&\tilde q_{kl}\delta t, &otherwise.
\end{aligned}
\right.
\end{equation}
In turn, the futures price   satisfies the recursive relation
\begin{equation}
F_i(t_n,x)=\sum_{j=1,\ldots,M}P_{ij}F_j(t_{n+},x),\label{futures_sum}
\end{equation}
where $F_j(t_{n+},x) = \lim_{t\downarrow t_n}F_j(t,x)$.

Since we assume Markov chain $\xi_t$ stay constant on time intervals $(t_n,t_{n+1}]$,  the martingale property of futures price  implies that
\begin{equation}
F_i(t_{n+},x)=\E^\Q[F_i(t_{n+1},X_{t_{n+1}})|X_{t_n}=x],
\end{equation}
for $i=1,\ldots,M$, and $n=0,\ldots,T/\delta t-1$. Therefore, we have following PDE for the futures price within each time interval $(t_n,t_{n+1}]$ and regime $i$,
\begin{align}
\partial_t F_i+\kappa_i(\theta_i-x)\partial_x F_i+\frac{\sigma^2_i}{2}\partial_{xx}F_i=0.\label{eqnfft}
\end{align}

Next, we apply Fourier transform to PDE (\ref{eqnfft}). To that end, we define
\begin{equation}
\F[f](\omega)=\int_{-\infty}^{\infty}f(x)\exp(-j\omega x)dx,\label{Fourier}
\end{equation}
where $j$ denotes the imaginary identity and $\omega$ denotes the frequency. Then, we obtain a  first-order PDE for $\hat{F}:= \F[F]$:
\begin{equation}
\partial\hat F_i+\kappa_i\omega\hat\partial_\omega F_i+\bigg(\kappa_i\theta_i j \omega+\kappa_i-\frac{\omega^2\sigma_i^2}{2} \bigg)\hat F_i=0,\label{eqnfft1}
\end{equation}
for $(t,\omega)\in[0,T)\times \mathbb{R}$, where $\hat F(t,\omega)\equiv \F[F](t,\omega)$. Noted that we use the following property of the Fourier transform to derive PDE (\ref{eqnfft1}),
\begin{equation}
\F[xf_x]=-\F[f]-\omega\F_\omega[f].
\end{equation}
To solve PDE (\ref{eqnfft1}) we employ the method of characteristics to get  
\begin{align}
\hat F_i(t_n+,\omega) = \phi_i(\delta t,\omega)\hat F_i(t_{n+1},e^{\kappa_i\delta t}\omega),\label{FFT_Solution}
\end{align}
with 
\begin{equation}
\phi_i(\delta t,\omega)=\exp\bigg(\kappa_i\delta t-\theta_i j\omega(1-e^{\kappa_i\delta t})+\frac{\sigma_i^2\omega^2}{4\kappa_i}(1-e^{2\kappa_i\delta t})\bigg).\notag
\end{equation}


Combining equation (\ref{futures_sum}) and (\ref{FFT_Solution}), we have following backward equation in the frequency space,
\begin{equation}
\hat F_i(t_n,\omega)=\sum_{j=1,\ldots,M}P_{ij}\phi_j(\delta t,\omega)\hat F_j(t_{n+1},e^{\kappa_j\delta t}\omega),\label{FFT_Backward}
\end{equation}
for $i=1,\ldots,M$, and $n=0,\ldots,T/\delta t-1$.
We then apply   backward induction via (\ref{FFT_Backward}) to calculate $\hat{F}_i(t,\omega)$. To recover the original futures price function, we apply inverse Fourier transform. For the numerical implementation of  Fourier transform and inverse Fourier transform, we utilize the standard fast Fourier transform algorithm. This FST method has been applied more broadly by \cite{JacksonJaimungalSurkov2008} to solve partial-integro differential equations (PIDEs) that arise in options pricing problems.
 For more details and other applications of this numerical approach, we refer to  \cite{SurkovThesis} and \cite{JacksonJaimungalSurkov2008}, and references therein. 
 


\subsection{Numerical Example}\label{Numerical_XOU}
We simulate the sample paths for spot price, futures prices, optimal investment, and optimal wealth process. The regime switching between two states, with transition probabilities  $q_{12}=2$ and $q_{21}=4$ which are entries of generator matrix $\bm{Q}$.  The trading horizon $\tilde{T}=0.6$, and the two futures contracts have maturities   $T_1=0.6$ and $T_2=0.8$. All parameters are summarized in the Table \ref{Parameters_Table}.\\

\begin{table}[H]
\centering
\begin{tabular}{c|c|c|c|c|c|cccccccc}
\hline
$\tilde q_{12}=q_{12}$& $\tilde q_{21}=q_{21}$& $\zeta(1)$& $\zeta(2)$&
$\kappa_1$& $\kappa_2$ & $\gamma$\\
\hline
$2$&$4$& 0.1 & 0.3 & 1 & 2 & 1\\
\hline
\hline

$\theta_1$& $\theta_2$&$\sigma_1$& $\sigma_2$&
$T_1$ & $T_2$ & $\tilde{T}$\\
\hline
2.5 & 2.7& 0.2& 0.3& 0.6& 0.8 & 0.6\\
\hline
\end{tabular}
\caption{Parameters for the RS-XOU model for Figures \ref{spot_path}--\ref{Phi_path}.} 
\label{Parameters_Table}
\end{table}
\vspace{5pt}

For the sample paths shown in Figure \ref{spot_path}, the market starts in regime  2, then switches  to regime 1  at time $t_1$) before returning to regime 2  at time $t_2$. The two price  levels $\exp(\theta_1)$ and $\exp(\theta_2)$ are the long-run means of the spot price  in regimes 1 and 2 respectively. In each regime, the spot price tends to  move towards the corresponding mean level. The spot and futures prices tend to move in tandem. Nevertheless, like in  Figure \ref{spot_path_GBM}, each regime switch can cause an instant jump in the  futures price but not in the spot price.  \\







Figure \ref{position_path} shows  the sample path of the optimal positions in  the two futures over the trading horizon. The optimal strategy is to long $T_1$-Futures and short $T_2$-Futures. The
long-short positions help reduce the effect of regime switching. As the market switches from regime 2 to regime 1 at time $t_1$, the magnitude of the futures position is immediately reduced.  The investor then take larger positions when the market returns to regime 2 from regime 1 at time $t_2$.  According to the sample paths, the optimal positions tend to decay in time  during each regime, meaning that the investor gradually reduces investments towards the end of the trading horizon. 
 \\

%



Figure \ref{wealth_path} shows the optimal  wealth process  over time.  By placing opposite position in two futures, we reduce the shock in portfolio vlaue due to  regime switching. Therefore, we do not observe  sizable jumps in the wealth process at the regime-swtiching times $t_1$ and $t_2$.  Just like the futures and spot prices, the   wealth process tends to decrease in regime 1 and increase in   regime 2 (towards the end of the trading horizon). \\


Lastly in Figure \ref{Phi_path}, we plot the sample path for $|\Phi|$, which is the absolute value of coefficient matrix determinant $\Gamma$. The strict positivity of $|\Phi|$ informs us that $\Gamma$ is invertible and thus the results in Section \ref{sect-utility} can be applied. Like in  Figure \ref{Phi_path_GBM}, $|\Phi|$ tends to increase exponentially in time.

\begin{figure}
\centering
\includegraphics[trim=0.5cm   0.4cm  1.5cm  1.7cm,clip,width=3.5in]{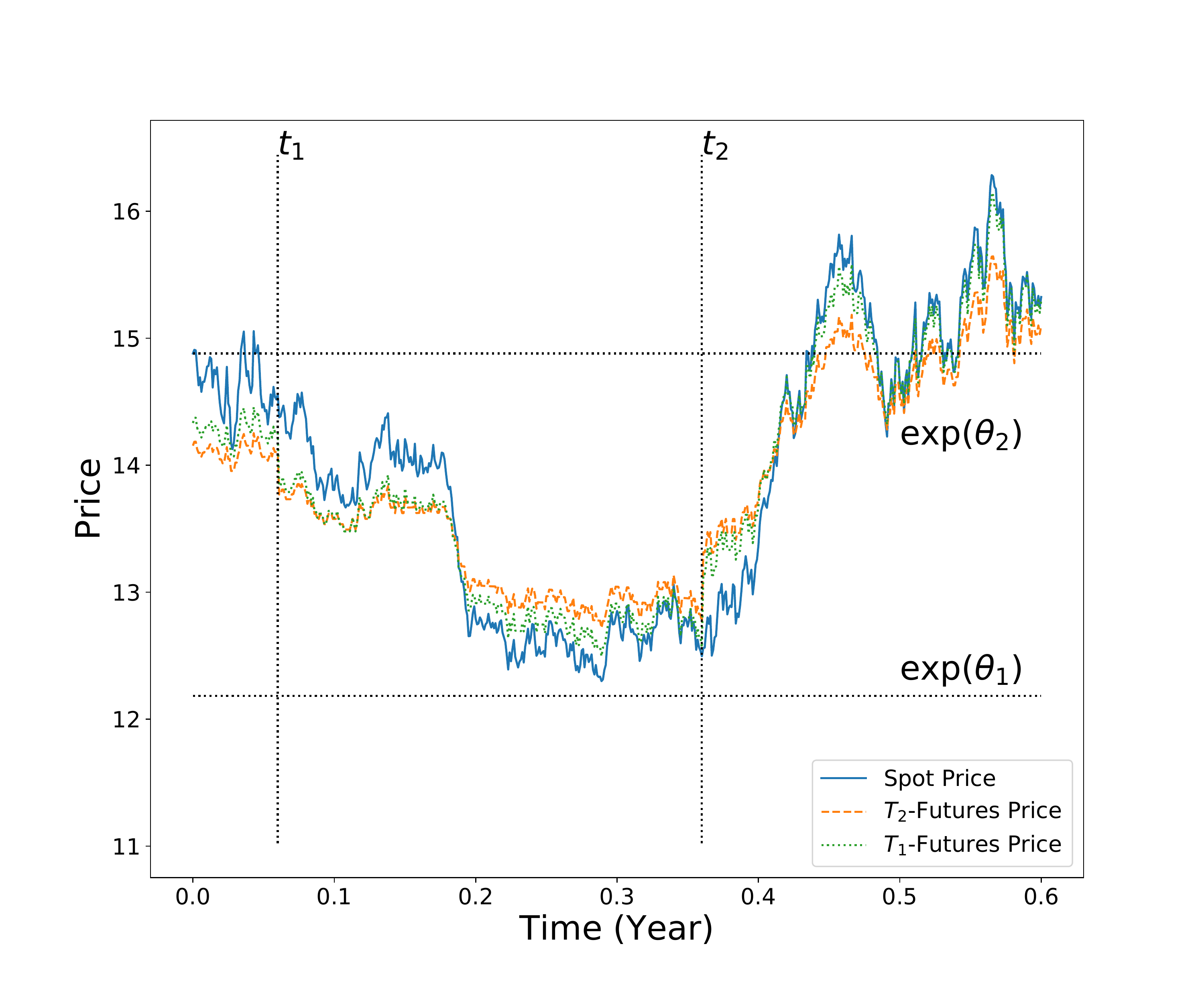}
\caption{Sample paths of the spot price, $T_1$-futures price, and $T_2$-futures price  over the trading horizon under the RS-XOU model. The market starts in regime 2, then switches to regime 1 at time $t_1$, before switching back to regime 2 at time $t_2$.   }\label{spot_path}
\end{figure}

%
%
%
%
\begin{figure}
\centering
\includegraphics[trim=0.5cm   0.4cm  1.5cm  1.7cm,clip,width=3.5in]{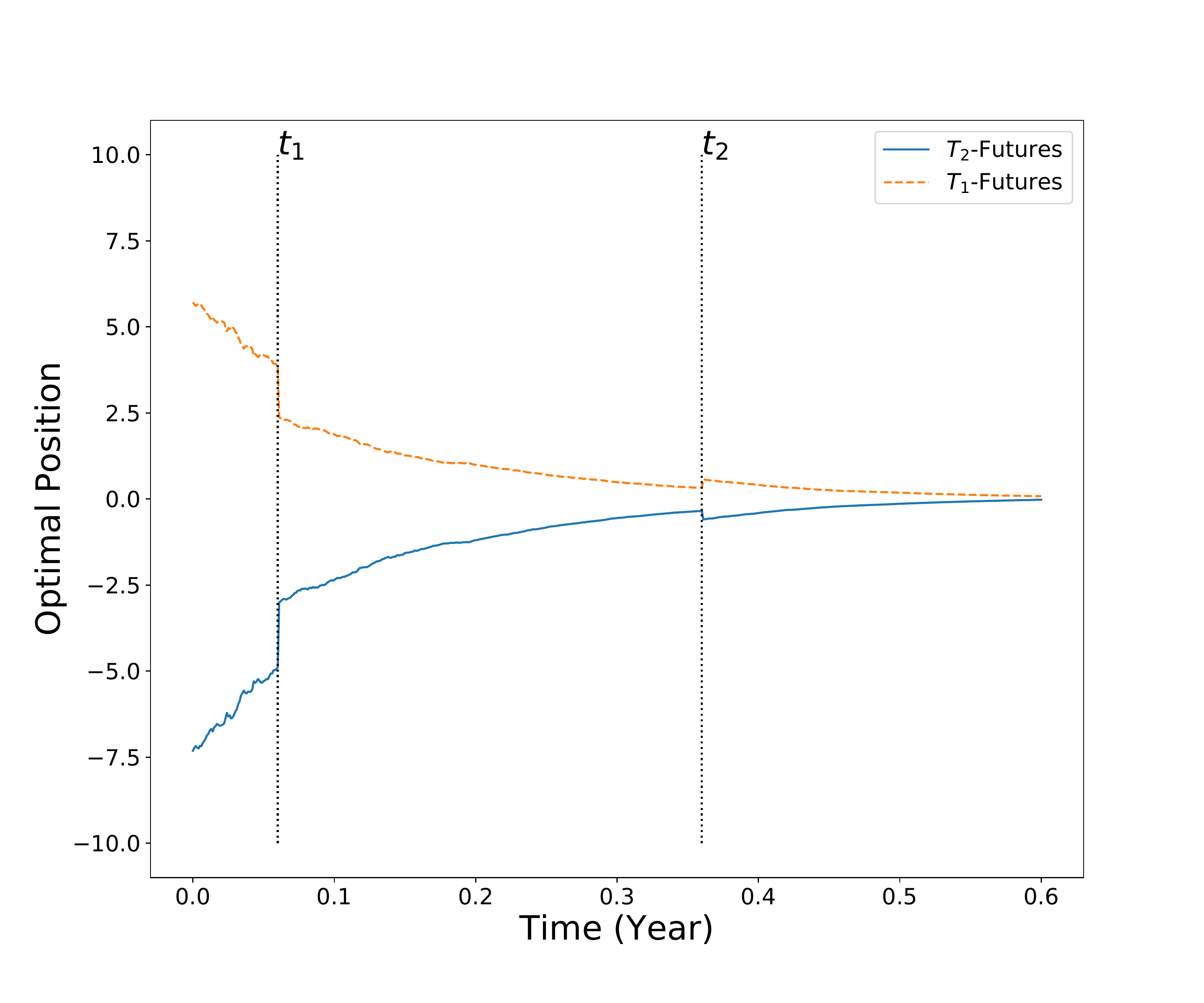}
\caption{Sample paths of the optimal positions ($\pi^{(1)},\pi^{(2)}$) in the  $T_1$-futures and $T_2$-futures respectively  under the RS-XOU model.  The futures positions tend to decrease over time and approaches zero near the end of the trading horizon. Jumps in both  positions occur at the regiem-switching times $t_1$ and $t_2$.  }\label{position_path}
\end{figure}

\begin{figure}
\centering
\includegraphics[trim=0.5cm   0.4cm  1.5cm  1.7cm,clip,width=3.5in]{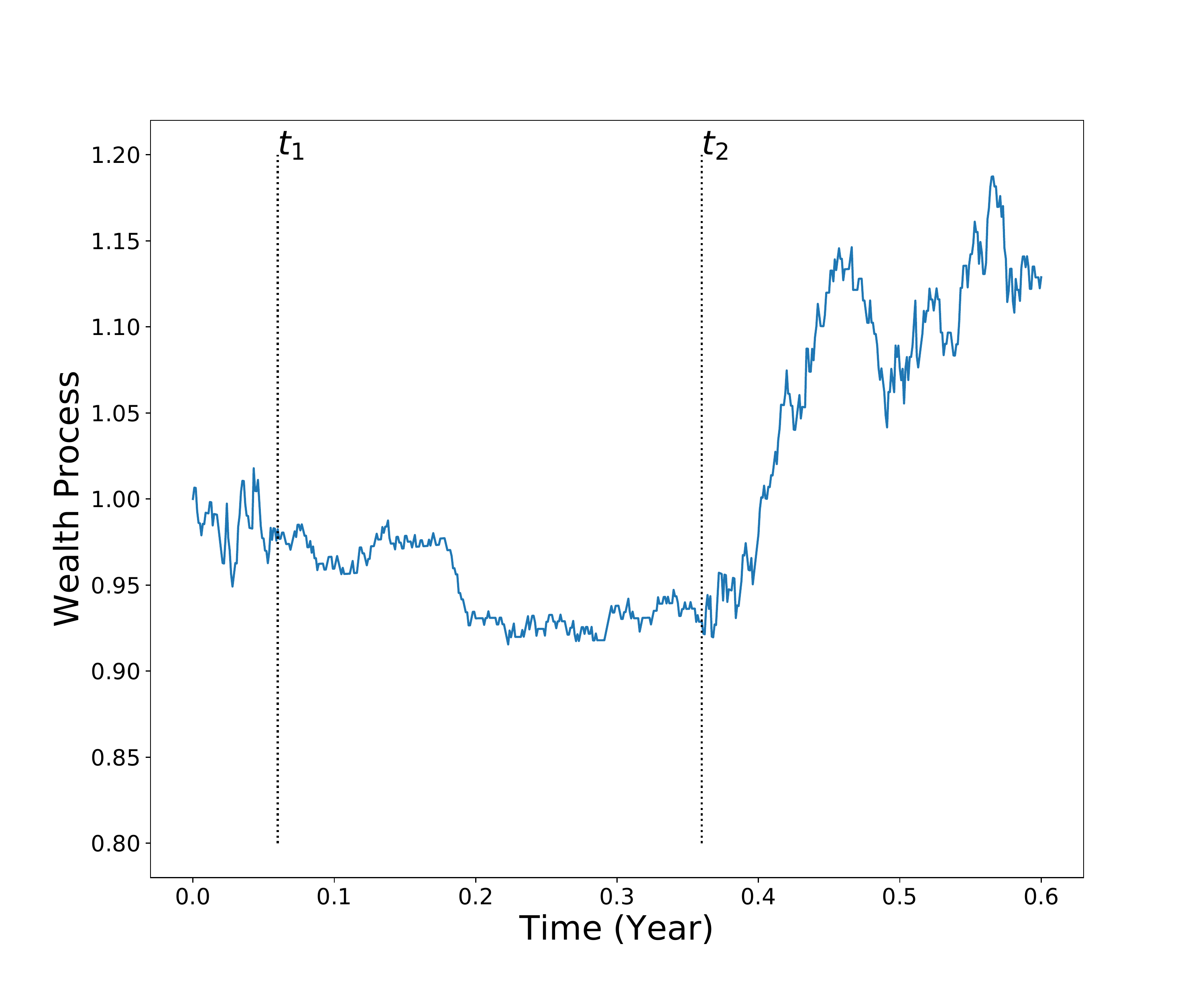}
\caption{Sample path of the optimal wealth process over the trading horizon under the RS-XOU model. }\label{wealth_path}
\end{figure}

\begin{figure}
\centering
\includegraphics[trim=0.5cm   0.4cm  1.5cm  1.7cm,clip,width=3.5in]{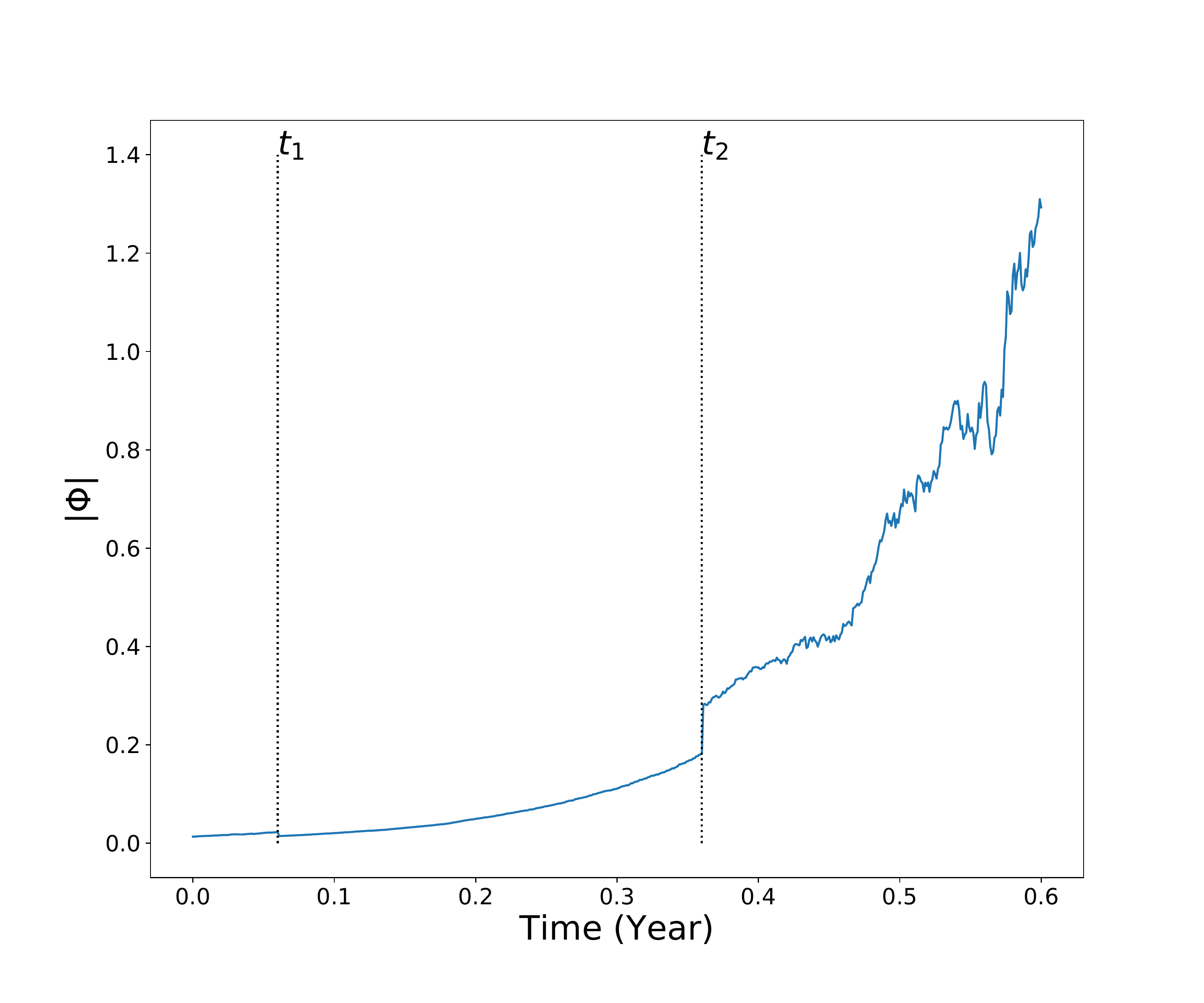}
\caption{Sample path of the absolute value of the determinant $\Phi$ for  the coefficient matrix $\Gamma$ under the RS-XOU model.}\label{Phi_path}
\end{figure}

\clearpage

\section{Conclusion}\label{sect-conclude}
We have analyzed  the problem of dynamically trading a portfolio of futures in a regime-switching market.  Under a general market framework, the portfolio optimization problem leads to the analytical and numerical studies of a system of HJB equations, which are solved explicitly. The optimal trading strategies and optimal wealth process are also given analytically and illustrated numerically.   Our  methodology has been applied to the  RS-GBM and RS-XOU models, but  it can also be used for other model specifications within the framework or adapted to study models with additional factors. Moreover, one can also utilize our framework to examine other classes of trading strategies, such as rolling and timing strategies.\footnote{See \cite{meanReversionBook} and \cite{LeungLiLiZheng2015} for such strategies   under mean-reverting models.} Other than futures portfolio, it is also useful to study the dynamic trading of other derivatives, such as options and swaps, in a regime-switching market.

\bibliographystyle{grappa}

\bibliography{mybib}

\begin{thebibliography}{}

\bibitem[\protect\astroncite{Angoshtari and Leung}{2019a}]{BahmanLeung}
Angoshtari, B. and Leung, T. (2019a).
\newblock Optimal dynamic basis trading.
\newblock {\em Annals of Finance}, 15(3):307--335.

\bibitem[\protect\astroncite{Angoshtari and Leung}{2019b}]{BahmanLeung2}
Angoshtari, B. and Leung, T. (2019b).
\newblock Optimal trading of a basket of futures contracts.
\newblock {\em Working Paper}.
\newblock Available at \url{https://ssrn.com/abstract=3467897}.

\bibitem[\protect\astroncite{Buffington and Elliott}{2002}]{Buffington2002}
Buffington, J. and Elliott, R.~J. (2002).
\newblock American options with regime switching.
\newblock {\em International Journal of Theoretical and Applied Finance},
  5:497--514.

\bibitem[\protect\astroncite{Chen et~al.}{2019}]{Chen19}
Chen, K., Chiu, M., and Wong, H. (2019).
\newblock Time-consistent mean-variance pairs-trading under regime-switching
  cointegration.
\newblock {\em SIAM Journal on Financial Mathematics}, 10(2):632--665.

\bibitem[\protect\astroncite{Elliott et~al.}{2005}]{Elliott2005}
Elliott, R.~J., Chan, L., and Siu, T.~K. (2005).
\newblock Option pricing and {E}sscher transform under regime switching.
\newblock {\em Annals of Finance}, 1(4):423-- 432.

\bibitem[\protect\astroncite{Jackson et~al.}{2008}]{JacksonJaimungalSurkov2008}
Jackson, K.~R., Jaimungal, S., and Surkov, V. (2008).
\newblock Fourier space time-stepping for option pricing with {L}\'evy models.
\newblock {\em Journal of Computational Finance}, 12(2):1--29.

\bibitem[\protect\astroncite{Leung}{2010}]{Leung_regime2010}
Leung, T. (2010).
\newblock A markov-modulated stochastic control problem with optimal stopping
  with application to finance.
\newblock In {\em Proceedings of the 49th {IEEE} Conference on Decision and
  Control}.

\bibitem[\protect\astroncite{Leung et~al.}{2016}]{LeungLiLiZheng2015}
Leung, T., Li, J., Li, X., and Wang, Z. (2016).
\newblock Speculative futures trading under mean reversion.
\newblock {\em Asia-Pacific Financial Markets}, 23(4):281--304.

\bibitem[\protect\astroncite{Leung and Li}{2016}]{meanReversionBook}
Leung, T. and Li, X. (2016).
\newblock {\em Optimal Mean Reversion Trading: Mathematical Analysis and
  Practical Applications}.
\newblock Modern Trends in Financial Engineering. World Scientific Publishing
  Company.

\bibitem[\protect\astroncite{Leung and Yan}{2018}]{LeungYan2018}
Leung, T. and Yan, R. (2018).
\newblock Optimal dynamic pairs trading of futures under a two-factor
  mean-reverting model.
\newblock {\em International Journal of Financial Engineering}, 5(3):1850027.

\bibitem[\protect\astroncite{Leung and Yan}{2019}]{LeungYan2019}
Leung, T. and Yan, R. (2019).
\newblock A stochastic control approach to managed futures portfolios.
\newblock {\em International Journal of Financial Engineering}, 6(1):1950005.

\bibitem[\protect\astroncite{Merton}{1971}]{Merton1971}
Merton, R. (1971).
\newblock Optimum consumption and portfolio rules in a continuous time model.
\newblock {\em Journal of Economic Theory}, 03:373--413.

\bibitem[\protect\astroncite{Surkov}{2009}]{SurkovThesis}
Surkov, V. (2009).
\newblock {\em Option Pricing Using Fourier Space Time-Stepping Framework}.
\newblock PhD thesis, University of Toronto.

\bibitem[\protect\astroncite{Zhou and Yin}{2003}]{ZhouYin2003}
Zhou, X. and Yin, G. (2003).
\newblock Markowitz's mean-variance portfolio selection with regime switching:
  A continuous-time model.
\newblock {\em SIAM Journal on Control and Optimization}, 42(4):1466--1482.

\end{thebibliography}

\end{document}